\documentclass[12pt]{amsart}
\usepackage[top=1.25in,bottom=1.25in,left=1.25in,right=1.25in]{geometry}
\usepackage[utf8]{inputenc} 
\usepackage{amsmath}
\usepackage{amssymb}
\usepackage{amsthm}
\usepackage{bbm}
\usepackage{graphicx}
\usepackage{float}
\usepackage{enumerate}
\usepackage{natbib}
\usepackage{setspace,lipsum}
\usepackage{bm}
\usepackage{booktabs}
\usepackage[table]{xcolor}
\usepackage{float}

\usepackage{mathrsfs}
\usepackage{comment}
\usepackage{apptools}
\usepackage{url}
\newcommand{\RN}[1]{  \textup{\uppercase\expandafter{\romannumeral#1}}}
\providecommand{\eps}{\varepsilon}
\usepackage{booktabs}

\usepackage[subrefformat=parens]{subcaption}
\usepackage{mathtools}
\usepackage{multirow}

\newtheorem{claim}{Claim}

\newtheorem{proposition}{Proposition}
\newtheorem{theorem}{Theorem}
\newtheorem{corollary}{Corollary}
\newtheorem{lemma}{Lemma}

\newtheorem{fact}{Fact}

\theoremstyle{remark}
\theoremstyle{remark}

\newtheorem{definition}{Definition}

\usepackage{etoolbox}
\usepackage[table]{xcolor}

\makeatletter
\patchcmd{\@maketitle}
  {\ifx\@empty\@dedicatory}
  {\ifx\@empty\@date \else {\vskip3ex \centering\footnotesize\@date\par\vskip1ex}\fi
   \ifx\@empty\@dedicatory}
  {}{}
\patchcmd{\@adminfootnotes}
  {\ifx\@empty\@date\else \@footnotetext{\@setdate}\fi}
  {}{}{}
\makeatother

\makeatother

\AtAppendix{
\numberwithin{equation}{section}
\numberwithin{definition}{section}
\numberwithin{theorem}{section}
\numberwithin{proposition}{section}
\numberwithin{lemma}{section}
\numberwithin{corollary}{section}}

\DeclareMathOperator*{\argmax}{arg\,max}

\DeclareMathOperator*{\marg}{marg}

\newcommand{\ltau}{\underline{\tau}}
\newcommand{\utau}{\overline{\tau}}
\DeclareMathOperator*{\sgn}{sgn}

\usepackage{amsthm}
\newtheorem{example}{Example}

\title{Friend or Foe: Delegating to an AI whose Alignment is Unknown}\thanks{We thank Yifan Dai, Sendhil Mullainathan, and Jean Tirole  for helpful comments, and NSF grants SES-2417162 and SES-2145352 for financial support.}
 \author{Drew Fudenberg}\thanks{Fudenberg: Department of Economics, MIT, drewf@mit.edu}  \author{Annie Liang}\thanks{
 Liang: Department of Economics, Northwestern University, annie.liang@northwestern.edu} 
\date{\today}

\begin{document}

\maketitle
\begin{abstract}

AI systems have the potential to improve decision-making, but decision makers face the risk that the AI may be misaligned with their objectives. We study this problem in the context of a treatment decision, where a designer decides which patient attributes to reveal to an AI before receiving a prediction of the patient's need for treatment. Providing the AI with more information  increases the benefits of an aligned AI but also amplifies the harm from a misaligned one. We characterize how the designer should select attributes to balance these competing forces, depending on their beliefs about the AI's reliability. We show that the designer should optimally disclose attributes that identify \emph{rare} segments of the population in which the need for treatment is high, and pool the remaining patients.

\end{abstract}
\bigskip
\bigskip
\bigskip
\bigskip
\bigskip
\thispagestyle{empty} \vspace{-1em}

\newpage \setcounter{page}{1}

\section{Introduction}
A key promise of artificial intelligence (AI) is its ability to refine and personalize recommendations.\footnote{This has been a leading application of machine learning techniques, see e.g., \citet{WagerAthey2018} and \citet{BryanKarlan2024BigLoans}.} For example, cardiac ablation,  a leading treatment for heart arrhythmia, is currently recommended based on the patient's age, arrhythmia burden, and a small set of additional summary statistics. In the future, AI may improve decision making by making  discoveries that permit better identification of who will benefit from treatment.

At the same time, there is increasing concern that AI systems may act with misaligned objectives relative to their designers, and that this may be difficult to prevent or detect because we do not fully understand the system's internal process. A natural safeguard is to validate the AI's predictions against established data. For example, if we know that cardiac ablations are more effective for patients under 40, then we can discipline the AI by verifying that its predictions align with these empirical findings. However, this constraint is  not enough to keep a misaligned AI from making misleading or harmful recommendations. We analyze  how much discretion to grant the AI and what information to reveal to it, given the designer's tradeoff between the risk from a misaligned AI and the benefits from an aligned AI.

To illustrate, suppose the designer  knows that a particular risky treatment will be successful for $40\%$ of the population and unsuccessful for the other $60\%$.  The designer gives the AI access to a binary covariate with unknown relation to need for treatment that also   splits the population into 40\% and 60\% groups,  and asks the AI to report their joint distribution. The designer then takes the optimal action for each patient given this reported distribution, where the designer receives a payoff of 1 from correctly treating a patient, a payoff of $-1$ from incorrectly treating a patient, and a  payoff of zero otherwise. What is the range of possible outcomes?

The best case is that the covariate is perfectly correlated with need for treatment, and the AI reveals this. Then the designer will treat precisely those patients who need treatment.  The worst case is that the covariate is perfectly correlated with need for treatment, but the AI instead reports that the conditional probability of need for treatment in the 60\% subgroup is $2/3$, while the conditional probability of need of treatment in the 40\% subgroup is zero. This is consistent with the designer's information that $40\%$ of the population needs treatment, so the designer cannot detect that the AI is falsely reporting. But this distribution leads the designer to treat everyone in the 60\% subgroup, none of whom actually need treatment. 

More generally, when the designer does not know whether the AI is aligned or misaligned, then they cannot adapt the choice of information to the AI's alignment. The designer must therefore trade off the possible gain from more accurate recommendations from an aligned AI with the possible loss from being misled by a misaligned one. We say that a pair of worst-case and best-case payoffs is \emph{implementable} if some choice of attributes leads to this worst case and best case. Among all such pairs, we characterize the efficient frontier---i.e., the implementable payoffs that are undominated by any other implementable pair---and explain how these efficient outcomes are implemented.

Section \ref{sec:Model} describes our framework.  A designer has access to historical data, which describes the probability that treatment is successful in each of several subgroups of the population. The designer chooses an additional set of attributes that further partitions each subgroup, where the link between those attributes and treatment outcomes is unknown. The AI is asked to report a joint distribution relating attributes to outcomes in each subgroup, where this distribution must agree with the designer's historical data and with the true joint distribution over covariates.\footnote{We consider it easier for the designer to verify the relationship between covariates in the population than to verify a relationship between those covariates and a treatment outcome.}

Depending on the true distribution and the AI's report, the designer's payoff falls in between two extremes: a best case,  where both the AI and Nature are benevolent and use the attributes to maximize the designer's payoff, and a worst case, where the AI and Nature are both malevolent and use the attributes to minimize it. We assume that the set of potential attributes is rich enough to allow for any desired finite-support joint distribution  between the original covariates and the additional ones,\footnote{ 
 For this reason, in the baseline model the physical identity of the additional attributes is irrelevant; their purpose is to allow the benevolent AI to provide useful advice.} and we define a frontier of undominated pairs of best-case and worst-case payoffs as we vary over different choices of attributes. Designers who aggregate these payoffs linearly have optimal points that are extreme points of this frontier, and we characterize those solutions. 
 
Section \ref{sec:Benchmark} considers a benchmark case in which priors are reasonably informative, so that every subgroup is low-risk ($<1/3$ need treatment) or high-risk ($>2/3$ need treatment).  We show that in this case the frontier is simply a line segment connecting a \emph{trust point} to a \emph{distrust point}. Distrust corresponds to giving the AI no information, so the best-case and worst-case payoffs coincide with the payoffs from following the best action under the prior. The trust point instead permits full-information targeting, where the designer precisely treats the patients who should be treated (as in our earlier example), but exposes the designer to a worst case  in which  treatment and need of treatment are negatively correlated. Moving along the frontier towards the trust point raises  the best-case payoff and lowers the worst-case  payoff at  a constant rate: the ratio between the value of correct treatment to the loss from an unneeded one.

The intermediate points on the frontier can be implemented by  attributes that isolate rare segments of the population. Intuitively, because the designer does not know how  attributes relate to treatment outcomes, they cannot ensure that the AI will use them correctly. The designer's only instrument for controlling mistreatment is to decide what fraction of the population the AI is allowed to influence. For example, when the designer knows that only $25\%$ of the population needs treatment, then giving the AI (only) an attribute that separates 1\% of the population from 99\% means that the AI can  induce mistreatment of at most 1\% of the population. Expanding the share of the population that the AI can influence simultaneously improves the best case and worsens the worst case.

Section \ref{sec:GeneralResults} turns to the general setting. When a subgroup is neither low risk nor high risk, the designer faces a worse tradeoff, because the AI has more flexibility over how to lie. For example, if 50\% of a population needs treatment  the AI can induce treatment on either the 1\% or 99\% subpopulations. This changes the frontier to two segments meeting at a new ``hedge point,'' with endpoints at the trust and distrust points (see Figure \ref{fig:intro}).

\begin{figure}
    \centering
    \includegraphics[scale=0.3]{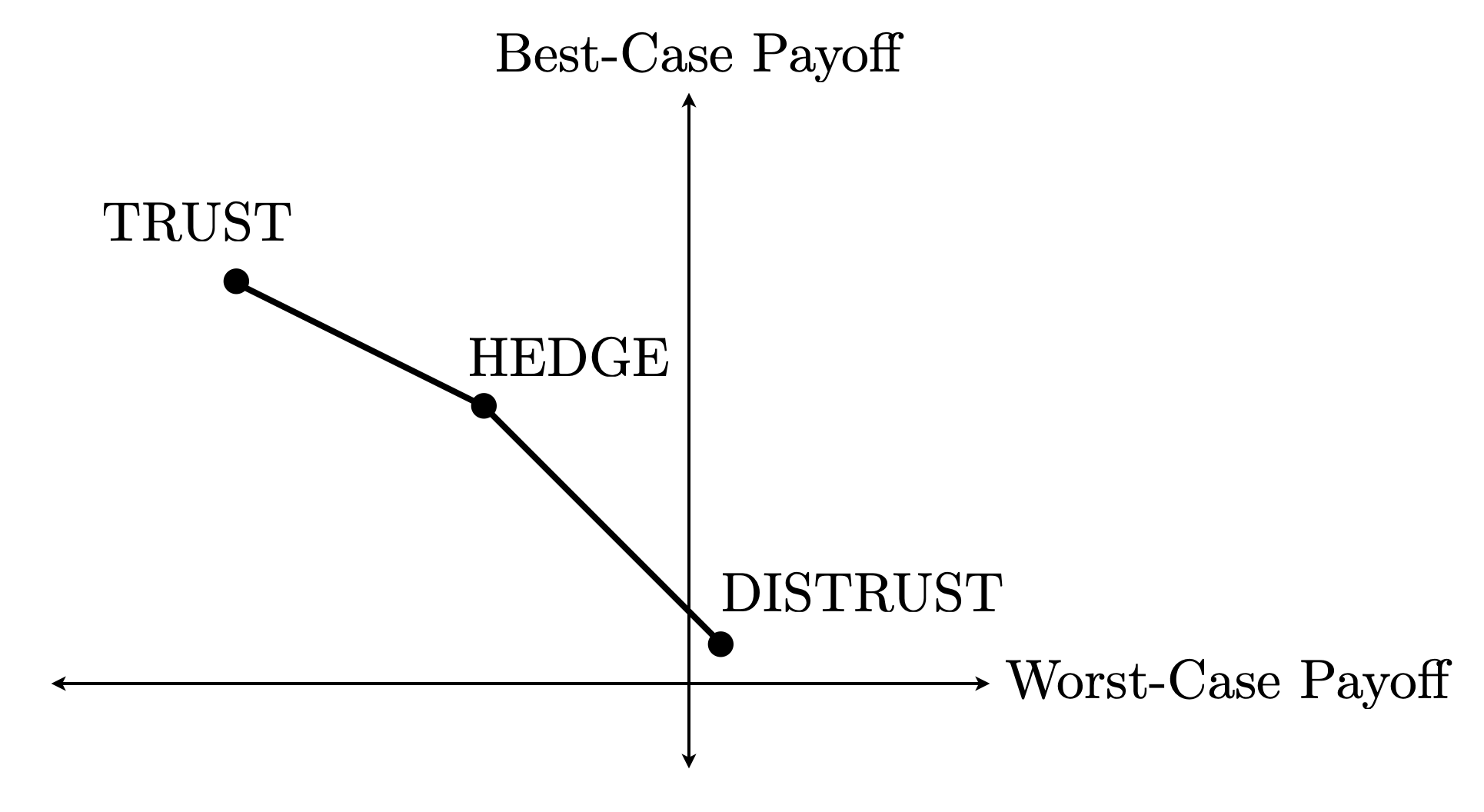}
    \caption{Depiction of the efficient frontier for subgroups with balanced need of treatment.}
    \label{fig:intro}
\end{figure}

A designer who is optimistic about the AI's alignment implements the trust point for every subgroup. As weight on the worst-case outcome increases, reliance on the AI is withdrawn subgroup by subgroup, starting with subgroups  where treatment outcomes are more skewed, and ending with subgroups with more balanced need of treatment (closer to 50-50). When the designer's weight on the worst case is sufficiently high, they implement the distrust point for every subgroup. Formally, the aggregate frontier is the  population-weighted Minkowski sum of the subgroup frontiers.

We have so far supposed that the designer believes that any relationship between the additional attributes and the need for treatment is possible. Section \ref{sec:ConstrainedNature}  extends our model to allow the  designer to choose bounds on how predictive the additional attributes are. We show that our original characterization of the frontier extends with some important differences, including that the tradeoff between the best-case and worst-case payoffs is now dictated by the bounds on informativeness. We further solve for the optimal informativeness bounds under an assumption that need for treatment in the population is relatively low (no more than  than $1/3$.) We show that a designer who puts sufficient weight on the worst case optimally chooses not to give the AI any information and a designer who puts sufficient weight on the best case  allows the  attributes to  be arbitrarily informative. For intermediate weights, the designer selects covariates that identify patients who definitely need treatment but pool together the remaining patients, so that evidence can be decisive for treatment but never against it.

\section{Related Literature}
\textbf{Empirical Motivation.
} In the last few years, the potential of AI to accelerate scientific discovery has become an active topic of discussion \citep{Patel2025AI2027,economist2023revolutionise_science}. AI systems have already made progress towards predicting the structure of proteins \citep{jumper2021alphafold,abramson2024alphafold3} and generating novel mathematical hypotheses \citep{davies2021advancing,romera2024funsearch}. \citet{Moor2023GMAI} suggest that AI will be capable of providing medical guidance across a diverse range of settings with very little task-specific data.

At the same time, the issue of misaligned AI has moved from a theoretical discussion of hypothetical problems to  a discussion of real-world examples of misalignment and how to prevent or reduce it. For example. \citet{greenblatt2024alignment} reports a large language model selectively complying with its training objective in training to prevent modification of its behavior out of training, \citet{park2024ai} reports that Chat GPT-4 deceived a human into solving an ``I’m not a robot'' task for it, and \citet{baker2025monitoring} reports several examples of ``chain of thought'' reasoning models that engage in \emph{reward hacking}, with the ``intent'' to subvert the task at hand.\footnote{The main focus of \citet{baker2025monitoring} is not documenting cases of reward hacking but rather training computer systems to detect it.} Moreover,  the consequences of misalignment seem likely to worsen as a AI models get more sophisticated \citep{dung2023current}. These findings suggest a need for theoretical frameworks studying the alignment problem and the tradeoffs between potential benefits and harms of AI.

\textbf{Methodology.}
Our model can be viewed a game with three players, namely the designer, the aligned AI, and the misaligned AI. The two types of AI can both send cheap-talk messages to the designer, as in \citet{krishna2001model}, but the game differs from standard cheap-talk models in several ways. First, the receiver (the designer) gets to pick the language or message space the receivers use. Second, while the messages themselves are not verifiable, they are subject to some consistency constraints: if the designer knows ex-ante that $1/3$ of the population should be treated, then the AI's advice can change \emph{which} patients are treated and even how many of them as long as it's consistent with $1/3$ of the population needing treatment.

Similarly, while some of our  arguments resemble arguments from the Bayesian persuasion literature following \citet{kamenica2011bayesian}, our model is quite different:  the senders lack commitment power, cannot provide verifiable signals, and can outright lie about the conditional need for treatment at some covariates. Within this literature, our model is especially related to the ``credible persuasion'' problem of \citet{lin2024credible}.  In \citet{lin2024credible}, the sender announces one information structure, and can deviate to any other information structure that leaves the message distribution unchanged; the announcement is said to be credible if the sender does not wish to make one of these unobservable deviations. The paper supposes there is a single type of sender, and shows that the sender's incentives must have the same modularity as the receiver's (with respect to the action and the state)  for credible persuasion to be possible.
In contrast, our model has  two types of senders, one of which is perfectly aligned with the receiver and one whose preferences are the exact opposite. In addition, while our aligned AI optimally chooses a credible information structure, our misaligned AI instead chooses a \emph{non-credible} information structure that maintains the same message distribution. 

Choosing a restricted set of covariates imposes constraints on the agent's feasible actions, as   in the delegation literature. In this sense, our model is somewhat related to work such as \citet{amador2006commitment}, where the period-$0$ self commits to second-period consumption restrictions even though this limits their ability to respond to taste shocks. Our model differs in in two main ways:  the designer's only instrument is the choice of covariate to let the AI use, and the AI does makes recommendations as opposed to decisions, so the recommendations must be incentive   compatible.

\textbf{Algorithmic Oversight.}  At a high level, our question of how rich a set of attributes to let the AI use is related to \citet{athey2020allocation}'s question of whether to delegate authority to a human or an AI, because ``delegating to a human'' is equivalent to not letting the AI use any attributes and basing the decision solely on the human's information. It is also vaguely related to the literature  on designing algorithms that reflect concerns for fairness and/or privacy, as here too the designer may be willing to sacrifice precision or accuracy by restricting the allowed covariates or adding noise to them.  See e.g. \citet{dwork2014algorithmic}, \citet{he2021private}, \citet{strack2024privacy}, and  \citet{liang2024algorithmdesignfairnessaccuracyfrontier}. 

\citet{yang2024explaining}  considers how to structure communication from an AI to a designer in a very different setting:  the designer is potentially unable to comprehend the true model, so the AI gives them a simpler explanation of it. The AI is known to be perfectly aligned with the designer, and chooses an explanation model that allows the designer to improve their 
worst-case payoff across the models that are consistent with any given explanation.

Economic theory has only recently started to study the issues posed by AI misalignment.   \citet{chen2024imperfect} studies delegation to a possibly misaligned AI. It proposes putting the AI in testing environments without revealing whether the task being performed is real or part of
a test. When the AI has imperfect recall, and the principal can conduct sufficiently many tests,  the principal attains the first best via screening (misaligned types eventually slip) and disciplining (they behave well to avoid detection). Our paper is complementary: rather than detecting misalignment through tests, we study how to structure information and discretion to bound the risks from a misaligned AI.

\vspace{-1em}
\section{Model} \label{sec:Model}
\subsection{Environment
} \label{sec:Environment} Let $\mathcal{Y} = \{0,1\}$ be a binary set of types and $\mathcal{A}=\{0,1\}$ be a binary set of actions. We  interpret $Y=1$ as meaning a treatment is effective and $A=1$ as a decision to treat, although the model applies more broadly. There is a human designer (they) and an AI agent (it).  The designer's payoff function is
\[u(a,y) = \left\{\begin{array}{cc}
1 & \mbox{if } (a,y) = (1,1) \\
0 & \mbox{if } a=0 \\
-1 & \mbox{if } (a,y) = (1,0)
\end{array}\right.\]
Thus action $a=0$ is ``safe'' while the payoff to action $a=1$ depends on the true type.

The designer's prior information takes the form of a tuple $I = (\mathcal{X}_0,\mu_0,(p_{x_0})_{x_0 \in \mathcal{X}_0})$, where $\mathcal{X}_0$ is a finite set of possible covariate values,
$\mu_0$ is the distribution over this set, and each $p_{x_0}$ is the conditional probability of $y=1$ at covariate $x_0$ (i.e., the fraction of patients with covariate value $x_0$ who need treatment). We will use $X_0 \sim \mu_0$ to denote the corresponding random variable. 

One can think of $I$ as an idealized data set that reports the treatment success rate for every covariate value in $\mathcal{X}_0$.
 For example, if $X_0$ describes only  the patient's age group, then $\mu_0$ is the population distribution over age groups, and $p_{x_0}$ is the   probability of treatment success for patients in age group $x_0$. We will sometimes refer to the patients with covariate value $X_0=x_0$ as the \emph{subgroup} $x_0$.

\subsection{Delegation to AI} \label{sec:Delegation} The designer considers using an AI agent to guide their decisions. Specifically, the designer augments $X_0$ with an auxiliary covariate $X_1$ that could help predict $Y$, while remaining uncertain about whether they do or in what way. The designer  then asks the AI to report a joint distribution $P$ for $(X_0,X_1,Y)$ and implements the decision rule that maximizes their  expected utility under $P$, i.e.,
\begin{equation} \label{eq:aStar}
a(x_0,x_1;P) = \left\{\begin{array}{cc}
1 & \mbox{if } \mathbb{E}_{P}(Y\mid (X_0,X_1)=(x_0,x_1)) \geq \frac{1}{2} \\
0 & \mbox{otherwise}
\end{array}\right.
\end{equation}
where we break ties in favor of treatment without loss (see Appendix \ref{app:TieBreak}).

We suppose that the designer has full and flexible design of the joint distribution of covariates, so that  $(X_0,X_1)$ can follow  any finitely-supported  joint distribution $\mu$ subject to the  constraint that $\marg_{\mathcal{X}_0}\mu = \mu_0.$ In this case we say that $\mu$ \emph{extends} $\mu_{0}.$ Our assumption is similar to the assumption in information design \citep{kamenica2011bayesian,BergemannMorris2019} that the sender can flexibly design an information structure. However, here we assume that the designer only chooses   the joint distribution of the covariates, and not how they relate to $Y$. 

\begin{example}
The decision is whether to recommend a cardiac ablation, and $X_0$ describes the patient's age bracket. The designer chooses additional attributes $X_1$ to give to the AI, which might include anything from the time pattern of the patient's arrythmias to the patient's dietary habits or recent travel. Our assumption of flexible design of $X_1$ means that the space of conceivable attributes is sufficiently rich that the designer can identify attributes with any desired relationship with $X_0$---for example, an attribute that is common among young patients and rare among older patients (such as a TikTok account).\footnote{Note that the designer can also construct new variables as garblings of existing ones.} The designer, however, does not know how these additional attributes correlate with need of treatment.
    
\end{example}

\subsection{The Designer's Ambiguity Set}  The designer knows the distribution of covariates, $(X_0,X_1) \sim \mu$, and also knows the conditional distribution of types given $X_0$,
but does not know the full joint distribution  $P^*$ of $(X_0,X_1,Y)$. We say that a distribution $P \in \Delta(\mathcal{X}_0 \times \mathcal{X}_1 \times \mathcal{Y})$ is \emph{admissible} if it is consistent with the designer's information. The set of all admissible distributions $P$ reflects the designer's uncertainty about just what the true distribution is.

\begin{definition} The joint distribution $P$ is \emph{admissible} if $P \in \Delta(\mathcal{X}_0 \times \mathcal{X}_1 \times \mathcal{Y})$ for some finite set $\mathcal{X}_1$ and:
\begin{enumerate}
    \item $\text{marg}_{\mathcal{X}_0 \times \mathcal{X}_1} P = \mu$
    \item $P(Y=1 \mid X_0=x_0) = p_{x_0}$ for every $x_0 \in \mathcal{X}_0$
\end{enumerate}
The designer's \emph{ambiguity set} $\mathcal{P}(I,\mu)$ is the set of all admissible distributions. \label{def:AmbiguitySet}
\end{definition}

This set is identical to the set of permitted distributions in \citet{lin2024credible}.\footnote{\citet{lin2024credible} considers a set of states $\Theta$, a set of messages $M,$ and an  information structure $\lambda \in \Delta(\Theta \times M)$, and  defines $D(\lambda) := \{\lambda' \in \Delta(\Theta \times M) :\lambda_\Theta' = \lambda_\Theta, \lambda'_M = \lambda_M \}$ to be the set of information structures that cannot be distinguished from $\lambda$ given the marginal distribution over states or messages. Our Definition \ref{def:AmbiguitySet} constructs the same set for the conditional distribution over $\mathcal{Y}$  and $\mathcal{X}_1$  given each $x_0 \in \mathcal{X}_0$. }

\subsection{Best and Worst Payoffs} \label{sec:Pareto}

 The AI is one of two possible types: it is either ``aligned'' and wants to help the designer, or ``misaligned'' and wants to hurt them. Both types of AI know the true $P^*$.\footnote{The AI that we consider doesn't necessarily rely on data to learn $P^*$; for example, this AI may have a deduced a theory of how $Y$ is determined. Nevertheless, one might think that learning $P^*$ is harder when $\mathcal{X}_1$ is richer. For most  of our results it will suffice to let $\# \mathcal{X}_{1}=2^{\# \mathcal{X}_0}$.}
 The aligned AI will always report the true $P^*$. The misaligned AI may not, but to avoid being detected,  its reported $P$ must  still be admissible.

Define
\[U(P,P^*) = \mathbb{E}_{P^*}[u(Y,a(X,P))]\]
to be the designer's expected payoff when he takes the optimal action under $P$ but the true distribution is $P^*$. The designer does not know the true distribution $P^*$, and also does not know whether the AI agent wants to help or hurt them.

To obtain clean characterizations, we focus on two extremes
\begin{align*} 
\underline{v}_I(\mu) & = \inf_{P^*\in\mathcal{P}(I,\mu)} \inf_{P\in\mathcal{P}(I,\mu)}
                 U(P,P^*) \\
\overline{v}_I(\mu) & = \displaystyle\sup_{P^*\in\mathcal{P}(I,\mu)} \sup_{P\in\mathcal{P}(I,\mu)}
                 U(P,P^*)
\end{align*}
which we call the \emph{worst-case} and \emph{best-case} payoffs. The payoff $\underline{v}_I(\mu)$ corresponds to pessimism both about Nature (who picks $P^*$) and an omniscient, misaligned AI  (who picks $P$ given knowledge of $P^*$). The payoff $\overline{v}_I(\mu)$ is analogously optimistic about both the AI and Nature when evaluating the best possible payoff. Our definition of $\overline{v}_I(\mu)$ above maintains symmetry with $\underline{v}_I(\mu)$, but it is equivalent to consider
\[\overline{v}_I(\mu) = \displaystyle\sup_{P^*\in\mathcal{P}(I,\mu)}
                 U(P^*,P^*)\]
                 where we directly assume that the aligned AI reports the true distribution $P^*$.
In Appendix \ref{app:AlternativePref}
 we consider an alternative specification of the best-case payoff motivated by  \citet{hurw1951}, in which the designer maximizes a weighted sum of the infimum and supremum over $P^*$.\footnote{Our analysis is  for one part of the parameter space.}

Because the designer does not know whether the AI is aligned or misaligned, they cannot tailor the choice of $X_1$ to the AI's alignment. Instead, they  must trade off the consequences of giving $X_1$ to the AI when it is used to help or harm the designer. Each choice of $(X_0,X_1)\sim\mu$ yields a payoff pair $(\underline{v}_I(\mu),\overline{v}_I(\mu))$ that  describes the designer's expected payoff in the worst and best case, respectively. Our main results  characterize the efficient frontier of payoff pairs $(\underline{v}_I(\mu),\overline{v}_I(\mu))$ that are generated by choices of $\mu$ and their randomizations.

\begin{definition}[Efficient Frontier] Let 
$C(I) = \mbox{conv}\{(\underline{v}_I(\mu), \overline{v}_I(\mu)) : \mu \mbox{ extends } \mu_0\}$
be the convex hull of feasible best-case and worst-case outcomes. The \emph{efficient frontier} given $I$ is
\[F(I) = \mbox{cl}\{(\underline{v},\overline{v}) \in C(I): \nexists  (\underline{v}',\overline{v}') \in C(I) \mbox{ s.t. }  (\underline{v}',\overline{v}')  \succ  (\underline{v},\overline{v}) \}\]
where $\succ$ 
denotes the usual dominance order, i.e., $(\underline{v}',\overline{v}') \succ (\underline{v},\overline{v})$ if $\overline{v}' \geq \overline{v}$ and $\underline{v}' \geq \underline{v}$ with at least one inequality strict. The frontier is the closure of the set of points that are undominated in the feasible set.

\end{definition}

The feasible set is convex, and the efficient frontier is part of its boundary. The extreme points of the efficient frontier are thus optimal for  designers whose preferences are a weighted sum of worst-case and best-case payoffs, i.e.
\[\alpha \underline{v}_I(\mu) + (1-\alpha) \overline{v}_I(\mu)\]
for $\alpha \in [0,1]$. One interpretation is that the designer is an expected utility maximizer, and $\alpha$ and $1- \alpha$ are the probabilities  the AI is misaligned or aligned. We will characterize these points and the covariates $X_1\sim \mu$ that implement them.

\begin{definition}[Implementation]
Say that $X_1 \sim \mu$ \emph{implements} $(\underline{v}_I,\overline{v})$ if
\[(\underline{v}_I(\mu),\overline{v}_I(\mu)) = (\underline{v}_I,\overline{v}).\]
Say that $X_1 \sim \mu$ \emph{limit-implements} $(\underline{v}_I,\overline{v})$ if there exists a sequence $\mu_n \rightarrow \mu$ such that
\[\lim_{n \rightarrow \infty} (\underline{v}_I(\mu_n),\overline{v}_I(\mu_n)) = (\underline{v}_I,\overline{v}).\]

\end{definition}

Our definition of implementation does not require unique implementation; that is, there may be multiple $X_1 \sim \mu$ that implement a given $(\underline{v}_I,\overline{v})$.

\section{Benchmark Case: Reasonably Informative Prior} \label{sec:Benchmark}

We first characterize the efficient frontier when the prior probability $p_{x_0}$ is either small or large---specifically, $p_{x_0}<1/3$ or $p_{x_0}>2/3$ in every subgroup $x_0$---so that the initial covariates $x_0$ are reasonably good  predictors of $y$. We show that the efficient frontier and optimal choices of $\mu$ take a particularly simple form, and can be characterized precisely in terms of three benchmark payoffs.

\begin{definition}[Default Action] Define
  \[d(p) = \left\{\begin{array}{cc}
  0 & \mbox{ if } p <1/2 \\
  2p-1 & \mbox{ if } p \geq 1/2.
  \end{array}\right.\] 
  This is the expected payoff when a fraction $p$ of the population needs treatment and the designer simply takes the best default action under the prior,  which is to treat everyone if $p\geq 1/2$ and otherwise not to treat anyone.
\end{definition}

\begin{definition}[Best Targeting] \label{def:BestSelection} Define
\[b(p,q) = \left\{\begin{array}{cc}
q & \mbox{ if } p \geq q \\
2p-q & \mbox{ if } p < q \end{array} \right.\]
 This is the expected payoff when  a fraction $p$ of the population needs treatment, and the designer treats the fraction $q$ of patients with (weakly) highest treatment need. If $q\leq p$ then all treated patients need treatment, while if $q>p$ then $\frac{p}{q}$ of these patients need treatment, leading to the payoff $q\left(2\cdot\frac{p}{q}-1\right) = 2p-q .$ The special case $b(p,p)$ corresponds to the \emph{full information payoff} that the designer receives by taking action $1$ if and only if $Y=1$. 
 \end{definition}

\begin{definition}[Worst Targeting] \label{def:WorstSelection} Define
 \[w(p,q) = \left\{\begin{array}{cc}
 -q & \mbox{ if } p \leq 1-q \\
 2p+q-2 & \mbox{ if } p > 1-q.
 \end{array} \right.\]
 This is the expected payoff when  a fraction $p$ of the population needs treatment and the designer selects the worst $q$-fraction of patients to treat (see Figure \ref{fig:countermonotonic} for a depiction when $q=p$). When $p+q<1$, none of the treated patients need treatment, yielding payoff $-q$. When $p+q>1$, then $\frac{p+q-1}{q}$ of the treated patients should be treated and    $\frac{1-p}{q}$   should not be, yielding payoff  $(p+1-q)-(1-p)=2p+q-2$.  
\end{definition}

If we view treatment $A$ and treatment need $Y$ as random variables, then the best selection payoff corresponds to $A$ and $Y$ being as positively correlated as possible (co-monotone), and the worst selection payoff corresponds to these variables being as negatively correlated as possible (counter-monotone). 

\begin{figure}[h]
\begin{center}
    \includegraphics[scale=0.25]{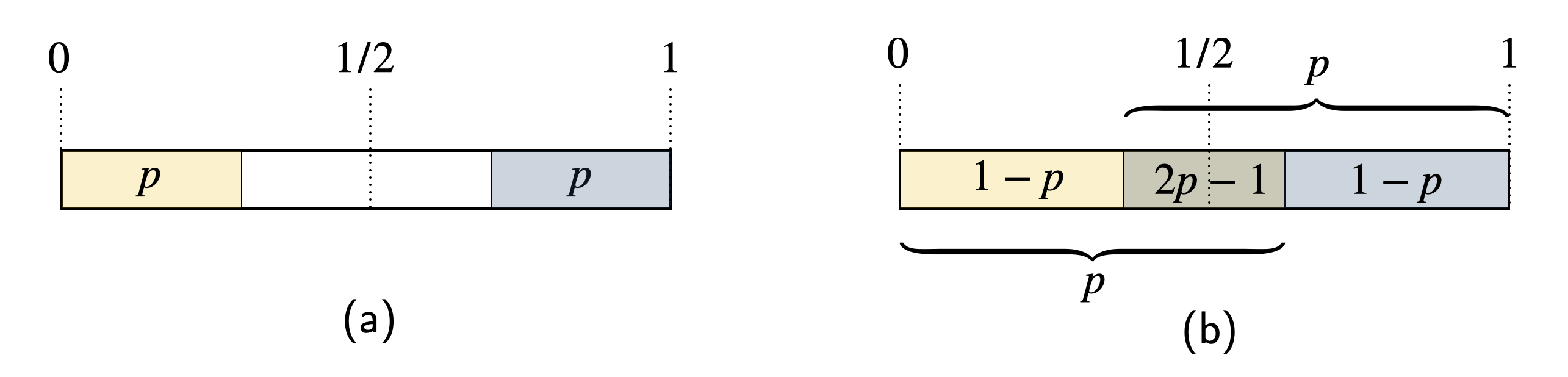}
    \caption{\small{ Yellow cells indicate patients who need treatment and blue cells indicate patients who are treated. \emph{Panel (a):} When $q=p\leq 1/2$, in the worst case the set of patients who are treated is  disjoint from the set of patients who are not treated; \emph{Panel (b):} When $q=p>1/2$, the sets of  patients who need treatment and who are treated under the AI's proposed rule must overlap. The worst-case corresponds to this overlap being as small as possible, in which case $2p-1$ patients are correctly treated and $1-p$ are incorrectly treated.}} \label{fig:countermonotonic}
    \end{center}
\end{figure}

These benchmarks define the endpoints of the efficient frontier.

\begin{proposition} \label{thm:Pareto} Suppose $p_{x_0} \notin \left[\frac13,\frac23\right]$ for every $x_0$. Then the efficient frontier is the line segment of slope $-1$ that connects the \emph{trust point} 
\[\mathsf{T}=\left(\sum_{x_0 \in \mathcal{X}_0} \mu_{x_0} \cdot w(p_{x_0},p_{x_0}), \sum_{x_0 \in \mathcal{X}_0} \mu_{x_0}\cdot  b(p_{x_0},p_{x_0})\right)\]  to the  \emph{distrust point}
\[\mathsf{D}=\left(\sum_{x_0 \in \mathcal{X}_0} \mu_{x_0} \cdot  d(p_{x_0}),\sum_{x_0 \in \mathcal{X}_0} \mu_{x_0}\cdot  d(p_{x_0})\right)\]
where $\mu_{x_0} := \mu_0(x_0)$ is the prior probability of $x_0$.
\end{proposition}

This frontier is depicted in Figure \ref{fig:ParetoLine}. Moving along the  frontier  towards the trust endpoint improves the best-case payoff by exactly the amount it lowers the worst-case payoff.

\begin{figure}[h]
\begin{center}
    \includegraphics[scale=0.32]{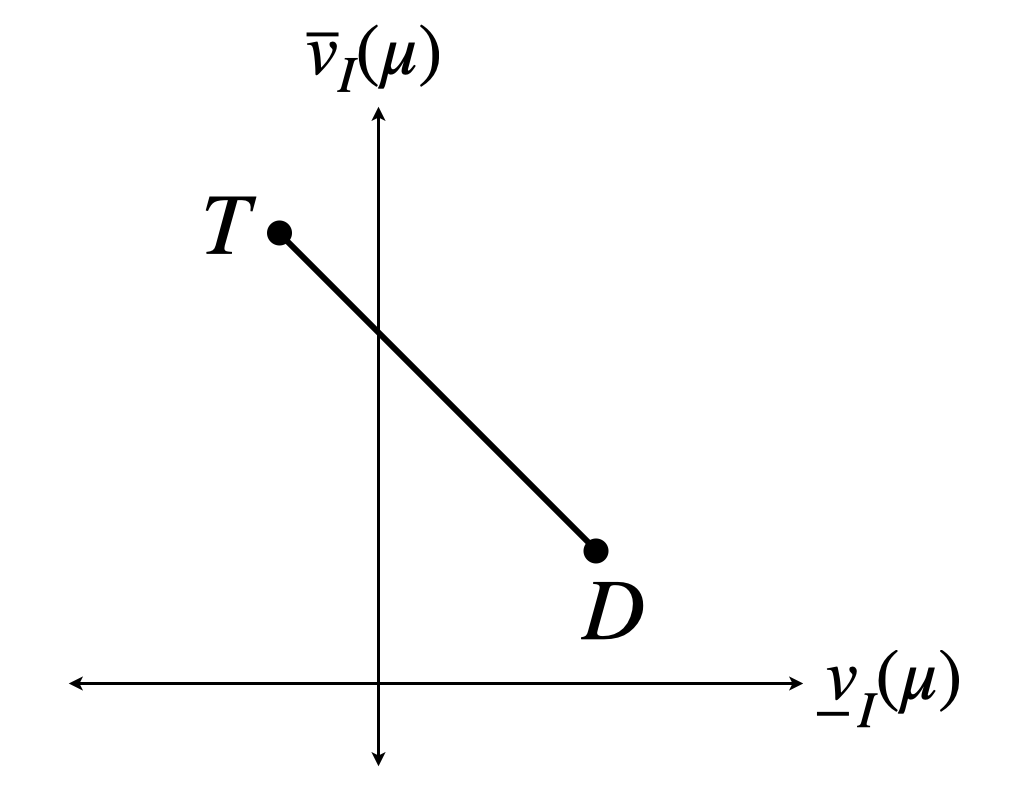}
    \caption{When $p_{x_0} \notin \left[\frac13,\frac23\right]$ for every $x_0$, then the efficient frontier is a line segment with slope $-1$.} \label{fig:ParetoLine}
    \end{center}
\end{figure}

The following proposition further explains how the trust and distrust points are implemented. 

\begin{proposition}\label{thm:support}  If  $p_{x_0} \notin \left[\frac13,\frac23\right]$ for every $x_0$, then:
\begin{enumerate}[(i)]
\item  If $ \alpha > 1/2 $, the distrust point $\mathsf{D}$ is optimal for the designer, and is implemented by any constant $X_1$. 
\item If $\alpha<1/2$, the trust point $\mathsf{T}$ is optimal for the designer, and is implemented by a Bernoulli random variable $X_1$ satisfying 
\[P(X_1 = 1 \mid X_0 = x_0) = p_{x_0} \quad \forall x_0 \in \mathcal{X}_0.\] 
\item If $\alpha=1/2$, any probability distribution over the solutions in (i) and (ii) is optimal.
\end{enumerate}

\end{proposition}

The distrust  point $\mathsf{D}$ corresponds to shutting the AI out: the additional covariate $X_1$ is fixed to a single value, so there is no scope for it to be used either for gain or harm. In this case, the designer's payoff is exactly what they would obtain from their prior information alone. The designer optimally implements this point whenever they put more weight on the worst-case payoff than the best-case payoff, i.e., when  $\alpha>1/2$.

At the opposite extreme, the  trust point  $\mathsf{T}$ allows maximum exposure to the AI. Here the designer chooses an $X_1$ that, in the best case, allows an aligned AI to precisely sort patients by whether they do or do not need treatment. The designer's best-case payoff is the full information payoff $b(p_{x_0},p_{x_0})$, achieved by treating only those patients who need treatment. However, such an $X_1$ exposes the designer to manipulation by a misaligned AI that delivers the countermonotonic payoff $w(p_{x_0},p_{x_0})$. The gap between the best-case payoff of $b(p_{x_0},p_{x_0})$ and the worst-case payoff of $w(p_{x_0},p_{x_0})$ reflects  the risk that must be accepted to benefit fully from a well-intentioned AI. The designer optimally implements this point whenever they put more weight on the best-case payoff than the worst-case payoff, i.e. when  $\alpha<1/2$. \\

To illustrate Proposition \ref{thm:Pareto} and gain intuition, consider the case in which $\mathcal{X}_0 =\{x_0\}$ is a singleton and let $p := p_{x_0}<\frac13$.

\begin{corollary} \label{corr:skewp} Suppose $p<\frac{1}{3}$. Then the efficient frontier is the line segment of slope $-1$ that connects the \emph{trust point} $T=(-p,p)$ to the distrust point $D=(0,0)$. This line can be parametrized as
$\{(-q,q) : q \in [0,p]\}$, 
where each point $(-q,q)$ is implemented by $X_1 \sim \mbox{Ber}(q)$. 
\end{corollary}

Intuitively, because the designer knows nothing about how $X_1$ should be used to predict $Y$, they cannot ensure that the AI uses it appropriately. The most they can do is constrain the overall frequency of treatment, since this bounds the potential damage from a harmful AI. 

If the AI is given the covariate $X_1 \sim \mbox{Ber}(q)$ where $q\in [0,p)$, then it can induce treatment for at most $q$ patients. If the AI is aligned, these are $q$ patients who in fact need treatment, and if the AI is misaligned, these are $q$ patients who don't need treatment. Thus the higher a weight the designer places on the worst-case outcome, the lower $q$ should be set.

What matters is not that $X_1$ is binary, but that it disaggregates at most a fraction  $q$ of patients. For example,
suppose a fraction $p=0.3$ of patients need treatment, and the designer is only willing to accept mistreatment of $10\%$ of the patients. Then the designer should choose an $X_1$ that isolates no more that share, either as one cluster of 10\% of patients or as several even rarer clusters that aggregate to at most 10\% of the population. Importantly, the AI  cannot induce treatment for the other 90\% of patients, since this would be inconsistent with the designer's prior knowledge that only 30\% of patients need treatment.

This final argument---that the AI cannot induce mistreatment for the majority of the population---breaks when  $p_{x_0} \in (1/3,2/3)$, as the following example illustrates.

\begin{example} \label{ex:EvenWorse}
    Let $\mathcal{X}_0 = \{x_0\}$ and let $p:=p_{x_0} = 2/5$ lie in the intermediate range $p \in (1/3,2/3)$ where the conditions of Propositions \ref{thm:Pareto} and \ref{thm:support} are not met.  
    
    If those propositions did apply, then a binary $X_1 \sim \mbox{Ber}(2/5)$ would yield the worst-case outcome of $w = -2/5$ (as part of the trust point $T=(-2/5,2/5)$). This corresponds to treating exactly $2/5$ of the population, none of whom need treatment.

    But this countermonotonic payoff is not the worst-case scenario. As shown in the introduction, by misreporting
    \[P(y=1 \mid X_1 =1)=0 \quad \mbox{ and } P(y=1 \mid X_1 = 0) = 2/3\]
    when in fact $X_1$ is perfectly correlated with $Y$, the AI can induce the designer to treat the $3/5$-mass of patients with $X_1=0$. Since none of these patients actually need treatment, the designer's payoff is $-3/5$.

  This worst-case payoff is only possible because $p > 1/3$, which permits the AI to induce treatment on either the event $\{X_1=1\}$ or $\{X_1=0\}$.\footnote{ If instead $p<1/3$,  the AI could not induce the designer to treat at $X_1=0$: this  would require  $P(Y=1 \mid X_1 =0)\geq 1/2$, implying 
    $P(Y=1) \geq \frac12 \cdot P(X_1=0) = \frac12(1-p) > p.$ }  This additional flexibility allows for more effective manipulation.\qed
 \end{example}

\section{General Case} \label{sec:GeneralResults}

We now characterize the efficient frontier and its implementation for arbitrary prior information. Section \ref{sec:Preliminaries} generalizes the definitions of the trust and distrust points  and introduces \emph{hedge points}. Section \ref{sec:Singleton} characterizes the frontier when $\mathcal{X}_0$ is a singleton. Section \ref{sec:FiniteX0} uses this to provide the full characterization. Section \ref{sec:ProofApproach} outlines the proofs of our main results.

\subsection{Preliminaries} \label{sec:Preliminaries} 
For each $x_0 \in \mathcal{X}_0$ define 
$T_{x_0}= (T^1_{x_0},T^2_{x_0})$
where
\begin{equation} \label{eq:GenTrust}
T^1_{x_0} = \left\{ \begin{array}{cc}
w(p_{x_0},p_{x_0}) & \mbox{ if } p_{x_0}\notin \left[\frac13,\frac23\right] \\
w(p_{x_0},1-p_{x_0}) & \mbox{ otherwise}
\end{array}\right. \quad \mbox{ and } \quad T^2_{x_0}= 
b(p_{x_0},p_{x_0}) 
\end{equation}
And for  each $x_0 \in \mathcal{X}_0$ define   
$D_{x_0} = (d(p_{x_0}),d(p_{x_0})).$  
The more general definitions of \emph{trust} and \emph{distrust points} are
\[\mathsf{T}= \sum_{x_0 \in \mathcal{X}_0} \mu_{x_0} \cdot T_{x_0} \quad \text{and}\quad 
\mathsf{D} = \sum_{x_0 \in \mathcal{X}_0} \mu_{x_0} \cdot D_{x_0}.\]
When $p_{x_0} \notin \left[\frac13,\frac23\right]$ for every $x_0$,  these points coincide with the trust and distrust points defined in the previous section. But when $p_{x_0} \in \left[\frac13,\frac23\right]$ for some $x_0$, then the worst-case payoff at that $x_0$ is not $w(p_{x_0},p_{x_0})$, but instead the lower payoff of $w(p_{x_0},1-p_{x_0})$ as in Example \ref{ex:EvenWorse}.

We define \emph{hedge points} as follows. Let $\mathcal{X}^\dag_0 = \{x_0 \in \mathcal{X}_0 :p_{x_0} \in [1/3,2/3]\}$ denote those covariate values $x_0$ for which $p_{x_0}$ violates the conditions of our previous results. For each $x_0 \in \mathcal{X}^\dag_0$ define $H_{x_0} = (H^1_{x_0},H^2_{x_0})$ 
where
\[H^1_{x_0} = \left\{\begin{array}{cc}  w(p,1-2p) & \mbox{ if } p_{x_0} \leq \frac12 \\
w(p,2-2p) & \mbox{ if } p_{x_0} > \frac12
\end{array}\right. \quad H^2_{x_0} = \left\{\begin{array}{cc} b(p,1-2p)  & \mbox{ if } p_{x_0} \leq \frac12 \\
 b(p,2-2p) & \mbox{ if } p_{x_0} > \frac12 \end{array}\right.\]
As we will show, hedge points are optimal choices for designers who put intermediate weights on the worst-case payoff.

We first characterize the frontier when $\mathcal{X}_0$  is a singleton, and subsequently generalize to arbitrary finite $\mathcal{X}_0$.

\subsection{When $\mathcal{X}_0$ is a singleton} \label{sec:Singleton}

\begin{lemma} \label{lemm:p} Suppose $\mathcal{X}_0$ is a singleton.
\begin{itemize}
        \item[(a)]  If $p_{x_0} \notin  \left[\frac13,\frac23\right]$,  then the efficient frontier is the line segment connecting the trust point $T_{x_0}$ to the distrust point $D_{x_0}$.
    \item[(b)] If $p_{x_0} \in  \left(\frac13,\frac23\right)$, then the efficient frontier is the union of  two line segments: one connecting the hedge point $H_{x_0}$ to the distrust point $D_{x_0}$ and one connecting the trust point $T_{x_0}$ to the hedge point $H_{x_0}$.\footnote{In the knife-edge case $p_{x_0} \in  \left\{\frac13,\frac23\right\}$, the efficient frontier is the line segment connecting the hedge point $H_{x_0}$ to the distrust point $D_{x_0}$.}
\end{itemize}
\end{lemma}

\begin{lemma} \label{lemm:pImplement} Suppose $\mathcal{X}_0=\{x_0\}$ is a singleton. Then:
 \begin{itemize}
    \item[(a)] The trust point $T_{x_0}$ is implemented by $X_1 \sim \mbox{Ber}(p)$.
    \item[(b)] The hedge point $H_{x_0}$ is limit-implemented by $X_1 \sim \mbox{Ber}(1-2p)$ if $p\leq 1/2$ and by $X_1 \sim \mbox{Ber}(2-2p)$ if $p>1/2$.
    \item[(c)] The distrust point $D_{x_0}$ is implemented by any constant $X_1$.
\end{itemize}
\end{lemma}

Unlike the case $p_{x_0} \notin \left[\frac13,\frac23\right]$, when $p_{x_0} \in \left(\frac13,\frac23\right)$  the frontier consists of two distinct segments. The first segment $\overline{H_{x_0}D_{x_0}}$, running from the hedge point to the distrust point, is qualitatively similar to the earlier frontier. Along this segment, the designer again  faces a one-for-one tradeoff: each unit of gain in the best-case payoff comes at the cost of a unit loss in the worst-case payoff.
By contrast, the second segment $\overline{T_{x_0}H_{x_0}}$, which extends from the trust point to the hedge point, has slope
\[\sigma(x_0) = \frac{T^2_{x_0} - H^2_{x_0}}{T^1_{x_0} - H^1_{x_0}} = \left\{\begin{array}{cc}
\frac{1-3p_{x_0}}{p_{x_0}} & \mbox{ if } p_{x_0} \in \left[\frac13,\frac12\right] \\
\frac{3p_{x_0}-2}{1-p_{x_0}} & \mbox{ if } p_{x_0} \in \left[\frac12,\frac23\right.]
\end{array}\right.\]
Along this part of the frontier, each unit gain in the best case costs more than a unit in the worst case, since $-1 \leq \sigma(x_0) \leq 0$ for all $p_{x_0} \in \left[\frac13,\frac23\right]$. 

\begin{example}
 Consider the setting of Example \ref{ex:EvenWorse}, where   $p=2/5$ of the patients need treatment. Then the distrust point is $D:=D_{x_{0}} = (0,0)$, the hedge point is $H:=H_{x_{0}}=(-1/5,1/5)$, and the trust point is $T:=T_{x_{0}} = (-3/5,2/5)$, as depicted in Figure \ref{fig:Single}.
 
 \begin{figure}[h]
\begin{center}
    \includegraphics[scale=0.25]{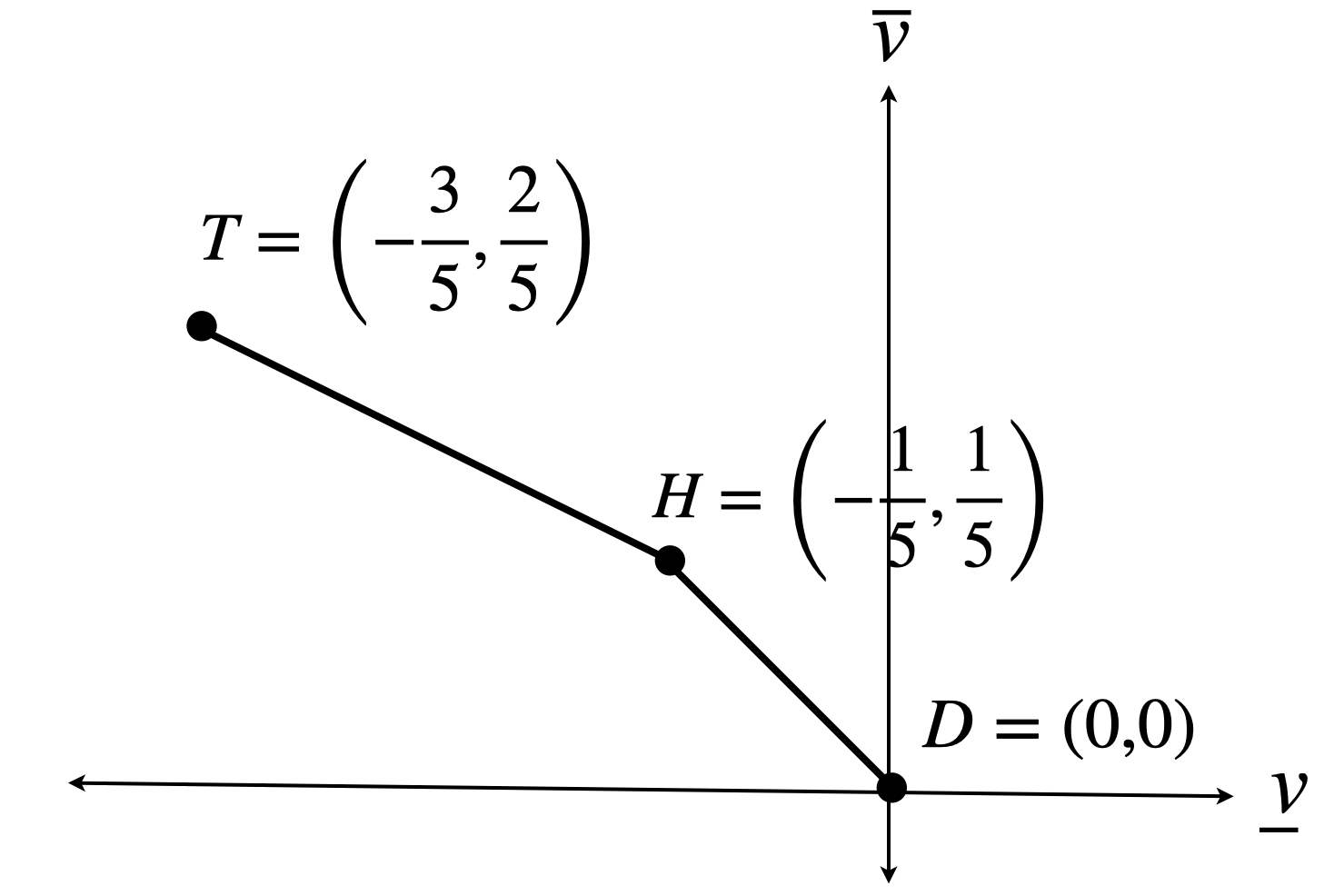}
    \caption{The frontier for Example \ref{ex:EvenWorse}.} \label{fig:Single}
    \end{center}
\end{figure}

The first segment from $D$ to $H$ can be parametrized as $\{(-q,q) : q \in [0,1/5]\}$. As in Section \ref{sec:Benchmark}, each point $(-q,q)$ is implemented by $X_1 \sim \mbox{Ber}(q)$, which leads both types of AI to recommend treatment with probability $q$. When the AI is aligned, these are patients who need treatment, and when the AI is misaligned, they are patients who do not. Increasing $q$ thus improves the best case and worsens the worst case by an equal amount.

But when  $q \in \left[\frac{1}{5},\frac{2}{5}\right]$, selecting $X_1 \sim \mbox{Ber}(q)$ gives the AI a choice between inducing treatment of  $q$ of the patients or  $1-q$ of the patients, since both are incentive-compatible. The best case still corresponds to correctly treating the $q$ of patients who need treatment, but the worst case becomes treating the $1-q$ of patients who don't need treatment. Thus the variable $X_1\sim \mbox{Ber}(q)$ for $q \in \left[\frac15,\frac25\right]$ implements the point $(-(1-q),q)$, and in particular $X_1 \sim \mbox{Ber}(2/5)$ implements the trust point $(-3/5,2/5)$.

Finally, the points on the line connecting the trust point to the hedge point are achieved through randomization, and (as we show in the proof of Lemma \ref{lemm:p}) all other feasible points are dominated by some point on one of these two line segments.
  
\end{example}

\subsection{Finite $\mathcal{X}_0$} \label{sec:FiniteX0}

For finite $\mathcal{X}_0$, the efficient frontier lies on the boundary of the Minkowski sum of the frontiers in Lemma \ref{lemm:p}, and takes the following form: Order the elements of $\mathcal{X}_0^\dag$ as $x_0^{(1)}, \dots, x_0^{(k)}$ where
\[\sigma(x_0^{(1)}) \geq \dots \geq \sigma(x_0^{(k)})\]
breaking ties arbitrarily. In this ordering,  the segment connecting $T_{x^{(j)}_0}$ to $H_{x^{(j)}_0}$ is flattest at $j=1$ and becomes progressively steeper as $j$ increases; equivalently,  the marginal cost of improving the worst-case outcome  in subgroup $x_0^{(j)}$ increases in $j$.  For convenience, write $\sigma_j := \sigma(x_0^{(j)})$ and let $S_j =\left\{x_0^{(1)}, \dots, x_0^{(j)}\right\}$ denote the first $j$ covariates in this ordering.

For each $j=1,\dots,k$, define the $j$-th hedge point to be
\begin{align*}
\mathsf{H}^{(j)} & = \sum_{x_0 \in S_j}  \mu_{x_0} \cdot H_{p_{x_0}} + \sum_{x_0 \notin S_j}  \mu_{x_0} \cdot T_{p_{x_0}}.
\end{align*}
To obtain $\mathsf{H}^{(j)}$, the designer implements the hedge point $H_{x_0}$ for each subgroup $x_0 \in S_j$, and the trust point $T_{x_0}$ for every other subgroup. Intuitively, $\mathsf{H}^{(j)}$ represents the outcome obtained by hedging on the $j$ subgroups where hedging is cheapest (the flattest segments) and trusting on the others.

\begin{definition} For any two points $A,B\in \mathbb{R}^2$ let $\overline{AB}$ denote the line segment connecting these points, i.e., $\overline{AB}=\{tA + (1-t)B: t\in [0,1]\}$.
\end{definition}

\begin{theorem} \label{thm:ParetoGeneral} The efficient frontier is
\[\overline{\mathsf{T}\mathsf{H}^{(1)}} \cup \overline{\mathsf{H}^{(1)}\mathsf{H}^{(2)}} \cup \dots \cup \overline{\mathsf{H}^{(k-1)}\mathsf{H}^{(k)}} \cup \overline{\mathsf{H}^{(k)}\mathsf{D}}\]
These are line segments whose slopes are given by $\sigma_1, \sigma_2, \dots, \sigma_k, -1$.
\end{theorem}

\begin{figure}[h]
\begin{center}
    \includegraphics[scale=0.25]{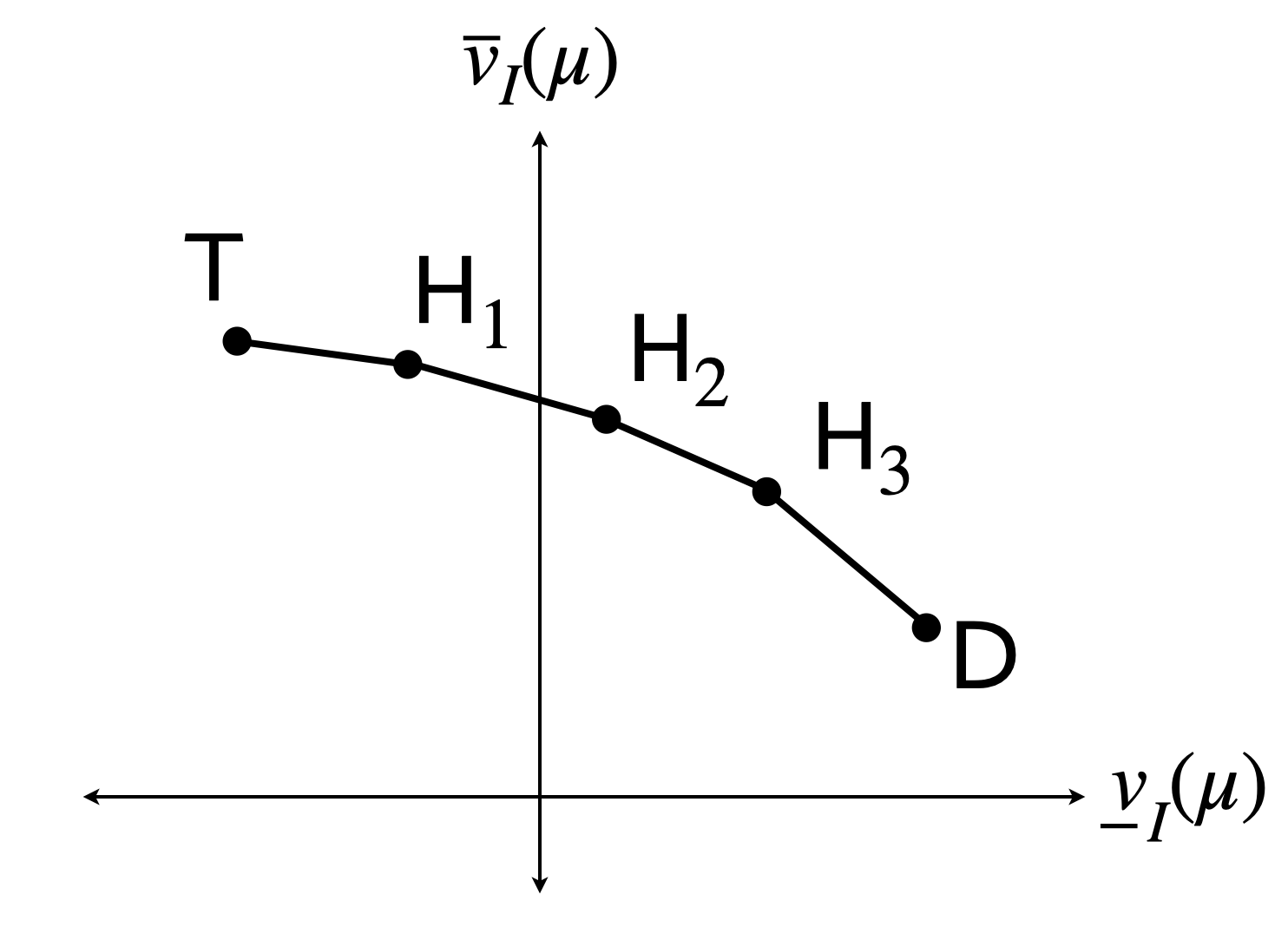}
    \caption{The frontier is a sequence of line segments connecting the trust point $\mathsf{T}$ to the distrust point $\mathsf{D}$.} \label{fig:ParetoLine}
    \end{center}
\end{figure}

This frontier is depicted in Figure \ref{fig:ParetoLine}. Its extreme points again correspond to the optimal points for designers who maximize $\alpha \underline{v}_I(\mu) + (1-\alpha)\overline{v}_I(\mu)$. Since the designer with utility weights $(\alpha,1-\alpha)$ has indifference curves that are linear with slope $-\frac{\alpha}{1-\alpha}$, we can determine which point is optimal by comparing this slope with the slopes of the line segments on the frontier. Specifically,  define $\alpha_j$ for each $j=1, \dots, k$ by 
\[-\frac{\alpha_j}{1-\alpha_j} = \sigma_j\]
and let $\alpha_{k+1} = 1/2$. These $\alpha_j$ are monotonically increasing in $j$ and satisfy $0\leq \alpha_j \leq 1/2$.

\begin{theorem} \label{thm:supportGeneral} Fix arbitrary utility weights  $(\alpha,1-\alpha)$. 
\begin{itemize}
\item[(i)] If $\alpha < \alpha_1$ then the trust point $\mathsf{T}$ is optimal, and is implemented by
\[X_1 \mid X_0=x_0 \sim \mbox{Ber}(p_{x_0}) \quad \forall x_0\]
\item[(ii)] If $\alpha \in [\alpha_j, \alpha_{j+1}]$ for $j=1,\dots,k$ then the hedge point $\mathsf{H}^{(j)}$ is optimal, and is limit-implemented by 
\[X_1 \mid X_0 = x_0  \sim \left\{\begin{array}{cc}
\mbox{Ber}(1-2p_{x_0})  & \mbox{ for all } x_0 \in S_j \mbox{ where } p_{x_0} < \frac12\\
\mbox{Ber}(2-2p_{x_0})  & \mbox{ for all } x_0 \in S_j \mbox{ where } p_{x_0} > \frac12\\
\mbox{Ber}(p_{x_0}) & \mbox{ otherwise}\end{array}\right.\]
\item[(iii)] If $\alpha > 1/2 $ then the distrust point $\mathsf{D}$ is optimal, and is implemented by any constant $X_1$.
\end{itemize}
\end{theorem}

Thus as the designer places increasing weight on the worst-case payoff, the optimal policy shifts gradually from full trust to full distrust, passing through these intermediate hedge points in a specific order. For small $\alpha$ (low weight on the worst case), the designer implements $\mathsf{T}$ by choosing the trust point $T_{x_0}$ for every subgroup $x_0$. As $\alpha$ grows larger, the designer successively switches from $T_{x_0}$ to $H_{x_0}$ for subgroups $x_0 \in \mathcal{X}_0^\dag$, beginning with those where the trust-hedge segment $\overline{T_{x_0}H_{x_0}}$ is flattest, i.e., those where hedging yields the best return. Equivalently, the designer ranks subgroups $x_0 \in \mathcal{X}_0^\dag$ by the distance between $p_{x_0}$ and $1/2$, switching first for those subgroups where $\vert \frac12 - p_{x_0}\vert$ is largest.\footnote{This is because $\sigma(x_0) > \sigma(x_0')$ if and only if $\vert \frac12 -p_{x_0} \vert > \vert \frac12 -p_{x_0'} \vert$. To see this, let $r := \vert 1/2 - p_{x_0}\vert$. Then $\sigma(x_0)  = \frac{3r-1/2}{1/2-r}$, which is increasing in $r$.}

Once all $x_0 \in \mathcal{X}_0^\dag$ are hedged (achieving the point $\mathsf{H}_k$), any further increase to $\alpha$ triggers a single, simultaneous move to full distrust $\mathsf{D}$. This is because the remaining frontier pieces $\overline{H_{x_0}D_{x_0}}$ all have slope $-1$, so if it is optimal to switch to $D_{x_0}$ for one subgroup, it is optimal for all of them. The following simple example illustrates the result.

\begin{example} A hospital predicts the probability that the patient needs an ablation given a baseline covariate vector $x_0$ describing age and frequency of irregular heartbeats. Suppose that $p_{x_0} \notin \left[\frac13,\frac23\right]$ for all except two $x_0$ subgroups, which we label $x_0^1$ and $x_0^2$. In these subgroups, the probabilities of treatment need are respectively $p_{x_0^1}=0.65$ (more skewed) and $p_{x_0^2} = 0.52$ (more balanced).

Then Theorem \ref{thm:ParetoGeneral} says that when $\alpha$ is low, the designer implements the trust point for each of these subgroups. As $\alpha$ increases, the designer first starts hedging for patients in subgroup $x_0^1$ by choosing $X_1$ to disaggregate a $\vert 2p_{x_0^1}-1 \vert = 0.3$ mass of patients in that subgroup. Thus at most 30\% of patients will be correctly treated (in the best case) or mistreated (in the worst case).

As $\alpha$ increases further, the designer also hedges for patients in subgroup $x_0^2$, this time by choosing $X_1$ to disaggregate a $\vert 2p_{x_0^2}-1\vert = 0.04$ mass of patients in that subgroup. Hedging here is more costly because it gives the AI almost no scope for influence. As $\alpha$ increases further, the designer finally switches to a constant $X_1$ for all subgroups, i.e., not relying on the AI at all.
\end{example}

\subsection{Proof Approach} \label{sec:ProofApproach}

First suppose $\mathcal{X}_0$ is a singleton,  and let $p:=p_{x_0}$ be the fraction of patients who need treatment (as in Lemmas \ref{lemm:p} and \ref{lemm:pImplement}). We show that the best-case and worst-case problems for a given $X_1 \sim \mu$ can be recast in the following way.

Let $Y^*$ and $Y$ be two binary-valued random variables, which we interpret as the true type variable and the AI's account of the type variable. (When the AI is aligned, always $Y=Y^*$.) The marginal distributions of $Y^*$ and $Y$ are fixed by admissibility: Each is Bernoulli with rate parameter $p$. The core of the problem lies in the coupling of the variables $(X_1,Y,Y^*)$, which leads to a joint law $\rho$. 

 We assume in the main proof that the designer always treats when indifferent, so  
\begin{equation} \label{eq:Action}
A = \mathbbm{1}\left(\rho(Y=1 \mid X_1)\geq 1/2\}\right)
\end{equation}
is the random variable corresponding to the designer's action,\footnote{Appendix \ref{app:TieBreak} shows that this tie-breaking assumption is without loss of generality.} and the designer's payoff is 
\begin{equation} \label{eq:Payoff}
\rho(\{A=1\} \cap \{Y^*=1\}) - \rho(\{A=1\} \cap \{Y^*=0\}),
\end{equation}
where $\rho(\{A=1\} \cap \{Y^*=1\})$ is the measure of patients that are correctly treated and $\rho(\{A=1\} \cap \{Y^*=0\})$ is the measure of patients that are incorrectly treated.  In the worst case, $(X_1,Y,Y^*)$ are coupled so as to minimize the designer's payoffs, and in the best case they are coupled so as to maximize the designer's payoffs.

Our proof proceeds in the following steps. 

\textbf{1) Best-case and worst-case payoffs for a given $A$.}  The designer's payoffs only depend on $(X_1,Y)$ insofar as they determine $A$. We thus first take the distribution of $A$ as given, setting aside the question of whether the AI could induce this random variable as the designer's action. It is straightforward to determine what coupling of $(A,Y^*)$ minimizes (or maximizes) the designer's objective in (\ref{eq:Payoff}): In the best case the true type variable $Y^*$ is maximally correlated with $A$, and in the worst case it is maximally negatively correlated with $A$. We derive the designer's payoffs in each case. (These are the $b(p,q)$ and $w(p,q)$ in Definitions \ref{def:BestSelection} and \ref{def:WorstSelection}, where $q=\mu(A=1)$.) 

\textbf{2) The feasible distributions for $A$.} We next ask which action distributions the AI can induce. Since $A \in \{0,1\}$, the variable $A$ is Bernoulli with some parameter $q$, and all constraints operate through $q$. Specifically, $q$ is feasible if there is a coupling $(X_1,Y)$ that yields $A \sim \mbox{Ber}(q)$, where $A$ is as defined in (\ref{eq:Action}). 

We decompose this condition into two parts. First, there must be  some  $\pi_1 \geq 1/2$ (corresponding to the posterior probability of need of treatment conditional on $A=1$) and some $\pi_0 \leq 1/2$ (corresponding to the posterior probability of need of treatment conditional on $A=0$) such that $q\pi_1 + (1-q)\pi_0 = p$, i.e., the posteriors aggregate to the prior $p$. If the designer could flexibly design $A$, then this would represent the entirety of the constraint on $q$. 

In addition,  the action $A$ must be measurable with respect to $X_1$. For example, if $X_1$ represents whether the individual is college-educated, and 37\% of the population is college-educated, then (given our deterministic tie-breaking rule) $q$ can only be either 0, 0.37, 0.63, or 1. We combine these constraints to characterize the set of feasible values for $q$ given any $p$ and $\mu$.

\textbf{3) The two AI's choices of $A$.} Among the feasible values of $q$, the aligned and misaligned AI each choose the one they want. This reflects the fundamental tension in our problem: the designer cannot adapt the design of $X_1$ to the nature of the AI, but must consider simultaneously what $X_1$ would allow each of these AIs to do. Continuing our college education example, it could be that the aligned AI's report induces the designer to treat only college-educated individuals, hence $q=0.37$, while the misaligned AI's report induces the designer to treat only individuals who are not college-educated, hence $q=0.63$.

The designer's best case for a fixed $X_1 \sim \mu$ can now be finally expressed as the payoff that results from the aligned AI choosing the feasible value of $q$ that maximizes the designer's payoffs when Nature is benevolent (as derived in (1)), and the worst-case payoff as the payoff that results from the misaligned AI choosing the feasible value of $q$ that minimizes the designer's payoffs when Nature is adversarial (as derived in (1)).

\textbf{4) Deriving the efficient frontier and generalizing.}
The above steps show how we determine the best-case and worst-case payoffs for a given $X_1 \sim \mu$. To derive the frontier, we need to consider all $X_1$ that lead to undominated payoffs. A key simplification is to show that it is sufficient to consider binary-valued $X_1$. Specifically, we show that  any richer covariate $X_1'$ can be replaced by some $X_1 \sim \mbox{Ber}(q)$ that yields (weakly) dominating payoffs. Our argument is by construction: We choose $q$ to be the action recommendation probability that the aligned AI would have chosen given the richer covariate $X_1'$. So the best-case payoff must be the same, but the worst-case payoff is weakly better since the misaligned AI is more constrained in which action distributions it can induce. 

The last step is to  derive the frontier by considering all possible Bernoulli variables $X_1$ and the resulting best-case and worst-case payoffs as derived in Step (1). We show that when $p<1/3$, the point implemented by $X_1 \sim \mbox{Ber}(q)$ is undominated if and only if $q\in [0,p]$, while if  $p\in(1/3,1/2]$, the point implemented by $X_1 \sim \mbox{Ber}(q)$ is undominated if and only if $q=p$ (implementing the trust point) or $q \in [0,1-2p)$. Similar statements hold for the mirrored case $p>1/2$. The efficient frontier thus simplifies to the convex hull of these implemented points, and we show that this is either a single line segment or the union of two line segments, as described in Lemma \ref{lemm:p}. We then finally derive the frontier when $\mathcal{X}_0$ is not a singleton  as a weighted Minkowski sum of the Pareto frontiers for each $x_0$. 

\section{Endogenous Informativeness Bounds} \label{sec:ConstrainedNature}

We have so far supposed that for any choice of $X_1$, the designer only knows  (1) the joint distribution of covariates $(X_0,X_1)$ and (2)  the conditional probability of treatment need given $X_0$. The designer's ambiguity set (see Definition  \ref{def:AmbiguitySet}) includes all joint distributions $P$ on $\mathcal{X}_{0} \times \mathcal{X}_{1}$ that are consistent with this prior information. In particular, the designer's posterior belief $P(Y=1 \mid X_0=x_0,X_1=x_1)$ can take any value in $[0,1]$. This reflects full ambiguity about how $X_1$ relates to outcomes. 

In practice, the designer may be able to place bounds on how much more information  a given $X_1$ provides than is in the original $X_0$. For example, the designer may not know whether giving the AI access to a patient's social media information  will increase or decrease the predicted need of treatment, but may believe that this information will have limited impact, i.e. that  the conditional probability $P(Y\mid X_1=x_1,X_0=x_0)$ would be close to $P(Y\mid X_0=x_0)$  for any realization $x_1$ of $X_1$. 

Section \ref{sec:Gentau} generalizes our previous setting by 
supposing  that  the designer can identify attributes $X_1$ to  not only satisfy any joint distribution of covariates $(X_0,X_1)\sim \mu$, but also to satisfy 
\[P(Y\mid X_1=x_1,X_0=x_0) \in [\ltau,\utau] \quad \forall x_1 \in \mathcal{X}_1\]
for given bounds $\tau=(\ltau,\utau)$, so that they are indirectly  choosing $\mu$ and $\tau$.  Section \ref{sec:Frontiertau} characterizes the frontier for a general set of bounds $\tau = (\underline{\tau},\overline{\tau})$ and shows how our original characterization of the frontier and its implementation extend.  Section \ref{sec:optimaltau}  describes the optimal choice of $\tau$ when treatment need in the population is not too large (specifically, less than $1/3$).

\subsection{Generalized Model} \label{sec:Gentau}

To simplify exposition we subsequently consider a singleton $\mathcal{X}_0$, understanding that  the construction of the frontier for general $\mathcal{X}_0$ proceeds as in Section \ref{sec:FiniteX0}: we take a weighted Minkowski sum of the frontiers for each individual $x_0$, and retain the undominated points. Henceforth let $p = P(Y=1)$ be the unconditional need of treatment in the population.  

The next definition generalizes  admissibility.

\begin{definition} For any $\tau = (\underline{\tau},\overline{\tau})$ where $0 \leq \ltau \leq p \leq \utau \leq 1$, the joint distribution $P$ is $\tau$-\emph{admissible} if $P \in \Delta(\mathcal{X}_1 \times \mathcal{Y})$ for some finite set $\mathcal{X}_1$ and:
\begin{enumerate}
    \item $\text{marg}_{\mathcal{X}_1} P = \mu$
    \item $P(Y=1) = p$ 
    \item $\ltau \leq P(Y=1 \mid X_1=x_1) \leq \utau$ for every  $x_1 \in \mathcal{X}_1$
\end{enumerate}
Let $\mathcal{P}_\tau(p,\mu)$ denote the set of all joint distributions that are $\tau$-admissible. This constitutes the designer's new ambiguity set.
\end{definition}

When $\ltau=0$ and $\utau=1$, $\tau$-admissibility reduces to our original definition of admissibility. Loosely speaking, more extreme values of $\ltau$ and $\utau$ correspond to covariates with higher potential predictive power. We allow the upper and lower bounds to be chosen separately rather than parameterizing informativeness by a single parameter, which allows us to consider covariates that provide more definitive evidence for need of treatment than for harm from treatment, or vice versa. For example, the designer might believe that some  patterns of arythmia episodes could be strong evidence that an  ablation is needed, but that no pattern would be decisive evidence that the patient does \emph{not} need an ablation. This would correspond to an ambiguity set where $\utau$ is much larger than the prior $p$, while $\utau$ is not much smaller than it. 

Now generalize the previous best-case and worst-case bounds to 
\begin{align*} 
\underline{v}_{p,\tau}(\mu) & = \inf_{P,P^*\in\mathcal{P}_\tau(p,\mu)} 
                 U(P,P^*) \\
\overline{v}_{p,\tau}(\mu) & = \sup_{P,P^*\in\mathcal{P}_\tau(p,\mu)} 
                 U(P,P^*),
\end{align*}
  and define the $\tau$-efficient frontier to be the pairs $(\underline{v}_{p,\tau}(\mu),\overline{v}_{p,\tau}(\mu))$ that are undominated within the corresponding feasible set.\footnote{That is, let $C_\tau(p) = \mbox{conv}\{(\underline{v}_{p,\tau}(\mu),\overline{v}_{p,\tau}(\mu)) : \mu \mbox{ extends } \mu_0\}$. Then the $\tau$-efficient frontier is 
$F_\tau(p) = \mbox{cl}\{(\underline{v}_I,\overline{v}) \in C_\tau(p): \nexists  (\underline{v}_I',\overline{v}') \in C_\tau(p) \mbox{ s.t. }  (\underline{v}_I',\overline{v}')  \succ  (\underline{v}_I,\overline{v}) \}$.}

\subsection{Efficient Frontier} \label{sec:Frontiertau}
We show that the $\tau$-efficient frontier qualitatively resembles the previous frontier, with a few key differences.  To avoid repetition, we present results here for $p < 1/2$ (with the mirrored case reported in Appendix \ref{sec:p>0.5}).

Regardless of what the $\tau$ bounds are, when the designer is sure the AI is misaligned, they withhold all information, so the distrust point remains  $D=(0,0)$. The role of $\tau$ becomes relevant when the AI is given some discretion. Then $\tau$ both limits the potential benefit and harm of the AI, and Definitions \ref{def:BestSelection} and \ref{def:WorstSelection} are modified as follows.

\begin{definition}[$\tau$-Best Selection] Define
\[b_\tau(p,q) = \left\{\begin{array}{cc}
q (2\utau-1) & \mbox{ if } p \geq q \utau + (1-q) \ltau \\
(2  (p-(1-q)\ltau)-q & \mbox{ if } p < q \utau + (1-q) \ltau 
\end{array}\right.\]
\end{definition}
The high-$p$ case corresponds to treating a fraction $q$ of patients, of which a fraction $\utau$ actually need treatment. This is feasible as long as $p$ (the unconditional probability that  treatment is needed) is large enough to support this belief. When $p$ is below the stated threshold, the fraction of true positives can be at most $\frac{p-(1-q)\ltau}{q}$, which yields the payoff in the second case. When $\underline{\tau}=0$ and $\overline{\tau}=1$,  $b_\tau(p,q)$ reduces to  $b(p,q)$.

\begin{definition}[$\tau$-Worst Selection]
Define
\[w_\tau(p,q) = \left\{\begin{array}{cc}
q(2\ltau-1) & \mbox{ if }  p \leq q \underline{\tau} +(1-q)\overline{\tau} \\

(2  (p-(1-q)\utau)-q  & \mbox{ if } p > q \underline{\tau} +(1-q)\overline{\tau} 
\end{array}\right.\]
When $p$ is low, the $\tau$-worst selection is to treat  fraction $q$ of patients, of which a fraction $\ltau$ actually need treatment. Otherwise if $p$ exceeds the stated threshold, the fraction of false positives can be at most $\frac{p-(1-q)\overline{\tau}}{q}$, yielding the second payoff.  When $\underline{\tau}=0$ and $\overline{\tau}=1$, then  $w_\tau(p,q)$ reduces to the original $w(p,q)$.
\end{definition}

We subsequently simplify notation by writing 
$w_\tau(q):=w_\tau(p,q)$ and $b_\tau(q):=b_\tau(p,q)$, 
where $p$ is the fraction of patients that need treatment in the fixed subgroup. The generalized trust point is  
\begin{equation} \label{eq:Ttau} T_\tau = \left\{\begin{array}{cc}
\left(w_\tau\!\left(q_T\right),b_\tau(q_T)\right) & \mbox{ if }p < p^{(1)}_\tau \\
\left(w_\tau\!\left(1-q_T\right),b_\tau(q_T)\right) & \mbox{ if } p^{(1)}_\tau \leq p \\
\end{array}\right.
\end{equation}
where 
\[q_T = \frac{p-\ltau}{\utau -\ltau}  \quad \quad \mbox{ and } \quad \quad p^{(1)}_\tau = \frac{\utau - 2\ltau^2}{1+2\overline{\tau}-4\ltau}.\] This definition parallels (\ref{eq:GenTrust}) with two changes. The first is that $q_T$ replaces $p$ as the treated fraction.  The quantity $q_T$ is defined to be the largest share of patients who can be assigned the probability $\utau$ for treatment need, given that the designer's posterior belief is bounded between $\ltau$ and $\utau$. In our original setting with $(\ltau,\utau)=(0,1)$, this largest share was $q_T=p$, since no more than a fraction $p$ of patients can need treatment with probability 1.

The second change is that $p^{(1)}_\tau$ replaces the previous threshold of $1/3$. Previously we observed that for all $ p\leq 1/3$, an AI given the covariate $X_1 \sim \mbox{Ber}(p)$ could only induce treatment for $X_1=1$, while for $p \in [1/3,1/2]$, the AI given this covariate could induce treatment on either $X_1=1$ or $X_1=0$. This new threshold $p^{(1)}_\tau$ has the same property; namely, when $p <p^{(1)}_\tau$ then the AI given a covariate $X_1 \sim \mbox{Ber}(q_T)$ can induce treatment for $X_1=1$ only, while if $p\in [p^{(1)}_\tau,1/2]$ then the AI can induce treatment for either $X_1=1$ or $X_1=0$.

The generalized hedge point is
$H_\tau= \left(w_\tau\!\left(q_H\right),b_\tau\!\left(q_H\right)\right)$
where $q^H= 
\frac{1-2p}{1-2\ltau}$. The following  results characterize the frontier and the implementation of its extreme points.

\begin{lemma} Suppose $p < 1/2$ and $\utau <\frac12$. Then the frontier is simply the distrust point $D_p$.
\end{lemma}

In this trivial case, the AI cannot induce the designer to assign probability greater than $1/2$ to need of treatment, so the designer never treats.

\begin{lemma} \label{lemm:pConstrained} Suppose $p < 1/2$ and $\utau \geq 1/2$. Recall that $p^{(1)}_\tau = \frac{\utau - 2\ltau^2}{1+2\overline{\tau}-4\ltau}$,  and let  $p^{(2)}_\tau = \frac{\ltau+\utau}{2}$. Then:
\begin{itemize}
    \item[(a)] If $p\notin \left[p^{(1)}_\tau,p^{(2)}_\tau\right]$, the efficient frontier is the line segment connecting the distrust point $D$ to the trust point $T_\tau$.
        \item[(b)] If $p \in \left(p^{(1)}_\tau,p^{(2)}_\tau\right)$, the efficient frontier is the union of the line segment connecting $D$ to $H_\tau$ and the line segment connecting $H_\tau$ to $T_\tau$.\footnote{If $p^*_1=p$, the efficient frontier is the line segment connecting the hedge point $H_\tau$ to the distrust point $D$.}
\end{itemize}
\end{lemma}

\begin{lemma} \label{lemm:pImplementConstrained} Suppose $p<1/2$ and $\utau\geq 1/2$. Then:
\begin{itemize}
    \item[(a)] The trust point $T_\tau$ is implemented by $X_1 \sim \mbox{Ber}\left(\frac{p-\ltau}{\utau -\ltau}\right)$.
    \item[(b)] The hedge point $H_\tau$ is limit-implemented by $X_1 \sim \mbox{Ber}\left(\frac{1-2p}{1-2\ltau}\right)$.
    \item[(c)] The distrust point $D$ is implemented by any constant $X_1$.
\end{itemize}
\end{lemma}

\begin{figure}
    \includegraphics[scale=0.25]{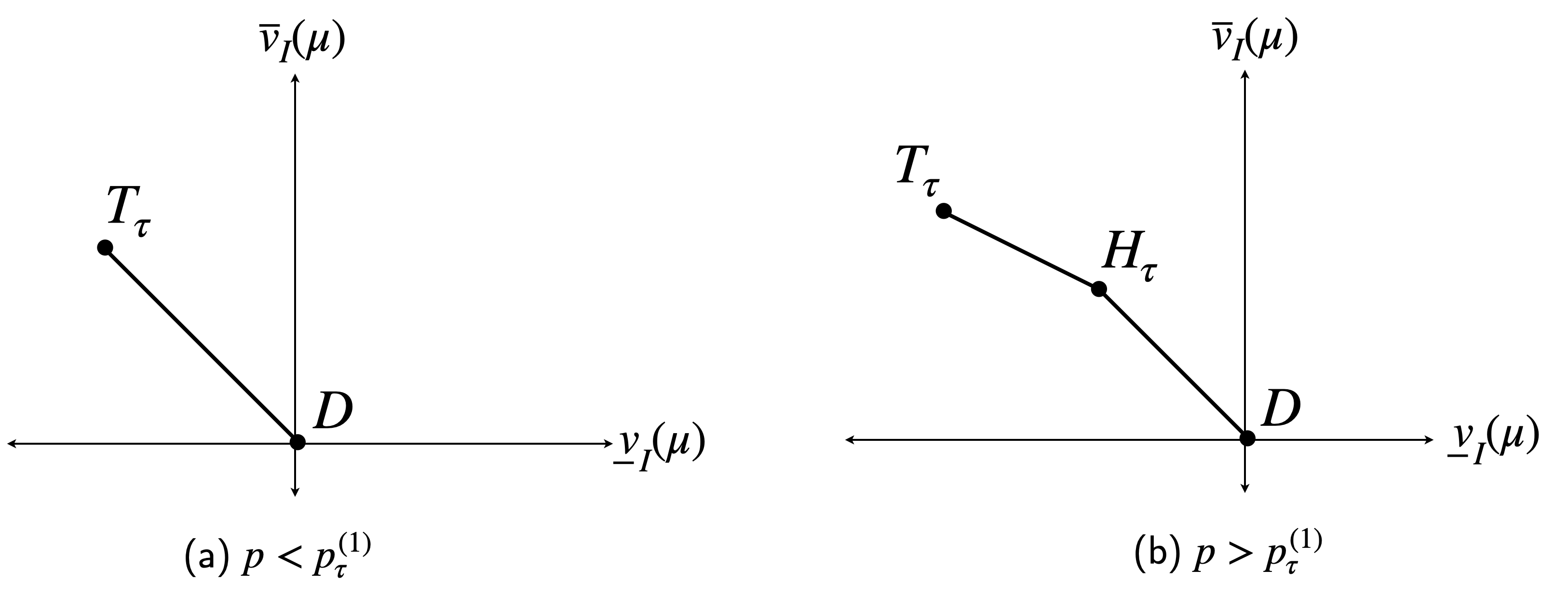}
    \caption{} \label{fig:taufrontier}
\end{figure}

This frontier is depicted in Figure \ref{fig:taufrontier}. As before, it is either a single line segment connecting the trust point to the distrust point, or two line segments with an intermediate hedge point.

To understand how $\tau$ influences the frontier, consider the simpler case in which it is a single line segment, as described in Part (a) of Lemma \ref{lemm:pConstrained}. When $\tau = (0,1)$ this line segment has  slope $-1$ (see Part (a) of Lemma \ref{lemm:p}). In   general, the  slope is $\frac{2\overline{\tau}-1}{2\underline{\tau}-1}$. When $\utau$ and $\ltau$ are symmetric around $1/2$---specifically, when $(\ltau,\utau)=(x,1-x)$ for some $0 \leq x \leq 1/2$---then the slope of this line segment is again $-1$. For these $(\ltau,\utau)$ pairs, the frontier is always a subset of the frontier described in Part (a) of Lemma \ref{lemm:pConstrained}, and the implementable part of that segment shortens as $x$ increases to $1/2$. Intuitively, tighter bounds reduce the AI's scope, much like coarsening  $X_1$ did in the previous section.

But when $\ltau$ and $\utau$ are asymmetric around $1/2$,  the tradeoff between the worst-case payoff and best-case payoff is no longer one-to-one. Specifically, when $\ltau+\utau >1$, the frontier is \emph{steeper} relative to when $\ltau+\utau=1$: each marginal increase in worst-case loss buys more best-case gain. Intuitively, this is because the maximum probability of treatment need lies farther above $1/2$ than the minimum probability lies below it, so the aligned AI has more power to effect a good outcome than the misaligned AI has to mislead the designer. Conversely, when $\ltau+\utau<1$, the lower bound is farther below $1/2$ than the upper bound is above it and the frontier flattens, reflecting a worse tradeoff.

The other impact of $\tau$ is on the trust point. At the trust point, the designer treats $\frac{p-\ltau}{\utau-\ltau}$ of the  patients. When the AI is aligned, a fraction $\utau$ of the treated patients truly need treatment, so the designer's payoff is $\left(\frac{p-\ltau}{\utau-\ltau}\right)(2\utau-1)$. This expression is increasing in $\utau$ and decreasing in $\ltau$: a more informative $X_1$ allows the aligned AI to generate a better payoff for the designer.

When these forces are considered together, a higher $\utau$ always benefits the designer: It improves the tradeoff between the worst-case and best-case, and it allows for a higher payoff when the AI is aligned. But there are two opposing effect to increasing $\ltau$: choosing a larger $\ltau$ improves the tradeoff between the worst-case and the best-case, but choosing a smaller $\ltau$ improves the best case payoff. In general the optimal choice of $(\ltau,\utau)$ can depend on the primitives in a complicated way, but the next section shows that the optimal choice admits a simple characterization when $p\leq 1/3$.

\subsection{Optimal choice of $\tau$} \label{sec:optimaltau}

Proposition \ref{prop:OptimalTau} characterizes the optimal choice of $\tau$ as a function of the designer's preference parameter $\alpha$, when the overall need for treatment $p$ is less than $1/3$.

\begin{proposition} \label{prop:OptimalTau} Suppose $p \leq 1/3 $. Define
$\ltau^\circ= 1-\sqrt{\frac{1-p}{2\alpha}}$ 
and
\[\underline{\alpha} (p) := \frac{1-p}{2} < \frac{1}{2(1-p)} =: \overline{\alpha}(p)\]
Then the optimal choices and utilities are given by:

\[\begin{array}{c|c|c|c|c}
\alpha \mbox{ } & \tau^* & \mu^* & u_1^* & \mbox{ extreme point} \\
\hline
\alpha < \underline{\alpha}(p) & (0,1) & X_1 \sim \mbox{Ber}(p)&  p(1-2\alpha) & \mbox{ trust} \\
\underline{\alpha}(p) < \alpha < \overline{\alpha}(p) & (\ltau^\circ,1) & X_1 \sim \mbox{Ber}\left(\frac{p-\ltau^\circ}{1-\ltau^\circ}\right)&  \frac{p-\ltau^\circ}{1-\ltau^\circ} (1-\alpha + \alpha(2\ltau^\circ -1)) & \mbox{ trust} \\
\alpha > \overline{\alpha}(p) & (p,1) & X_1 \mbox{ constant} & 0 & \mbox{ distrust} 
\end{array}\]
\end{proposition}

\begin{figure}[h]
\begin{center}
\includegraphics[scale=0.3]{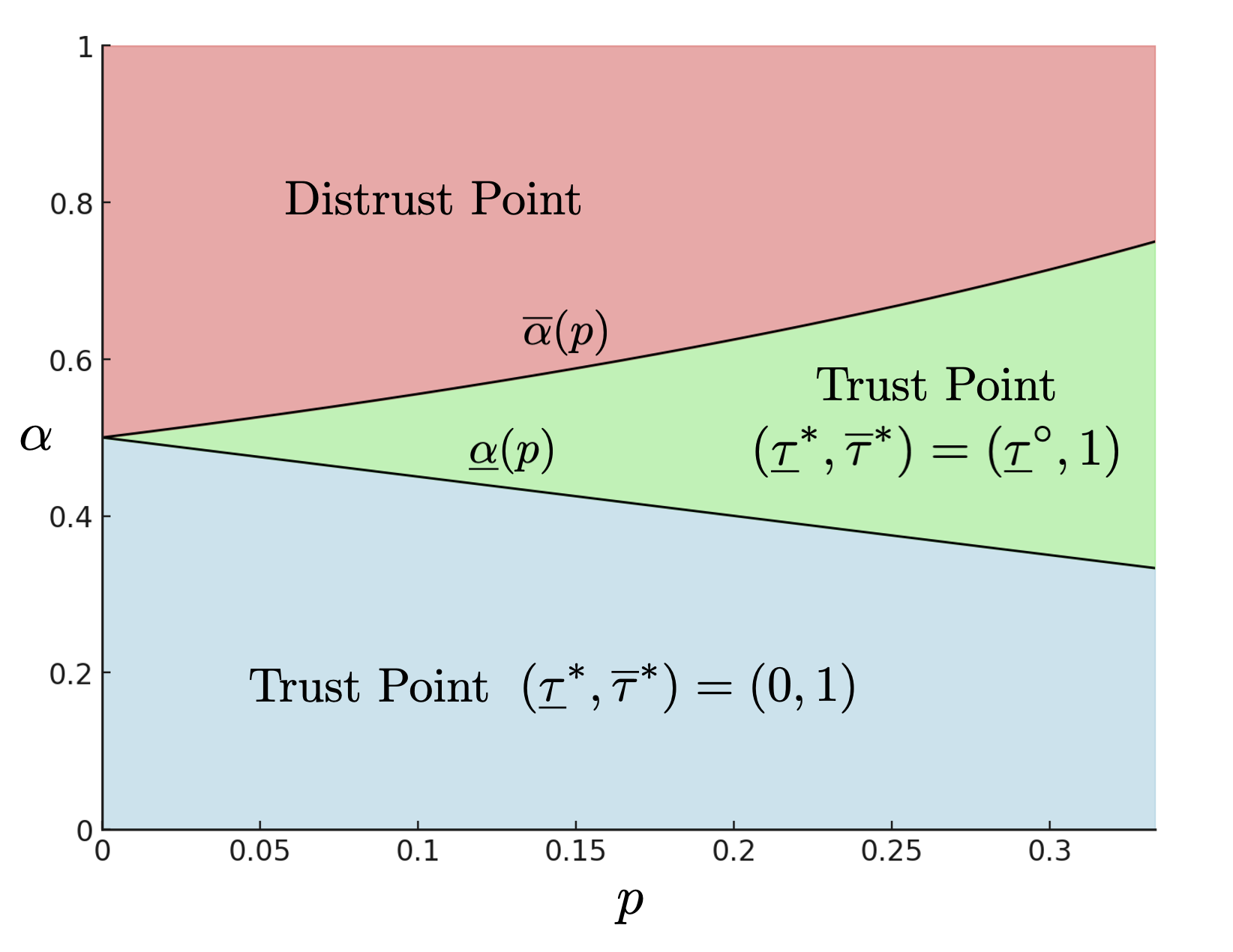}
\caption{Depiction of the optimal $\tau$ and extreme point when a fraction $p$ of patients need treatment and the designer places weight $\alpha$ on the worst-case payoff.}
\end{center}
\end{figure}

Unlike in the previous sections, here the frontier is itself endogenous. That is, the designer chooses both the frontier's shape (via $\tau$) and also the point to implement on it (via $\mu$). We show that no matter  the designer's $\alpha$, they  will always choose a $\tau$ that leads to the frontier taking the form of a single line segment from the trust point to the distrust point $(w_\tau(q_T),b_\tau(q_\tau))$. When $\alpha$ (the weight on the worst case) is sufficiently large,  the designer implements the distrust point by giving the AI no information. 

Otherwise, the designer implements the trust point, but the payoffs at this point vary depending on how the designer's  choice of $\tau$. When $\alpha$ is sufficiently small,  the optimal choice is $\tau=(0,1)$, corresponding to a covariate $X_1$ with maximum potential for informativeness. The optimal  distribution of this $X_1$ is Bernoulli with rate parameter $p$, so the solution in this case is the same as that given by Proposition \ref{thm:support} where    $\tau=(0,1)$.  For larger $\alpha$, the designer sets $\utau=1$ but chooses an interior $\ltau = 1-\sqrt{\frac{1-p}{2\alpha}}$ as a function both of the prior probability $p$ and also of the preference parameter $\alpha$. The distribution of $X_1$ is Bernoulli with rate parameter $\frac{p-\ltau^\circ}{1-\ltau^\circ}$. That is, the designer chooses a covariate that isolates a $\frac{p-\ltau^\circ}{1-\ltau^\circ}$ fraction of the population that needs treatment, with  the remaining patients needing treatment with probability $\ltau$. As $\alpha$ grows larger, the fraction of treated patients grows smaller, and the probability of need of treatment for the remaining patients grows larger. 

\begin{figure}[h]
    \includegraphics[scale=0.23]{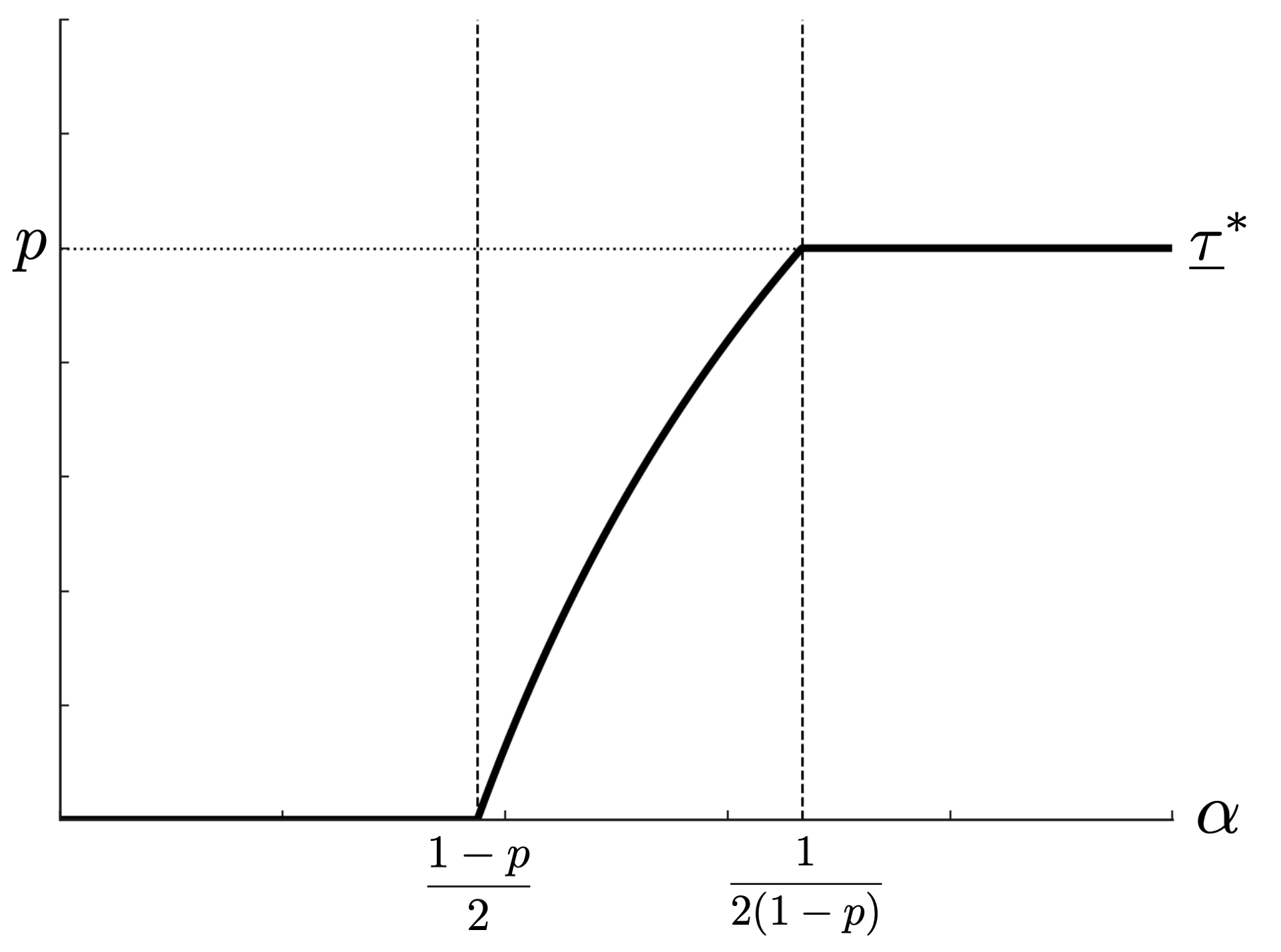}
    \caption{Plot of the optimal choice of lower bound $\ltau^*$}
\end{figure}

\section{Conclusion}

Our analysis leaves open several interesting questions for future work. First, this paper  considers a completely aligned or completely misaligned AI. It would be  interesting to formalize ``partial alignment'' and explore what its consequences would be.  Second, we conduct our analysis in the simplest possible decision setting: a one-time decision of which patients to treat. Future work could consider repeated interaction with an AI and more general state spaces and payoff functions. Third, we suppose that Nature's choice of distribution (i.e., the true relationship between attributes and treatment need) is as favorable or unfavorable as possible or (in Appendix \ref{app:AlternativePref}) a lottery between these two extremes. Future work could relax this assumption by  constraining the space of conceivable distributions given other observables, for example by supposing that our informativeness bounds $\ltau$ and $\utau$ are a function of the richness of the covariate space $\mathcal{X}_1$.

\appendix
\section{Preliminaries}

\begin{lemma} \label{lemm:Restate} The worst-case payoff is \begin{align*} \label{eq:RestatedWorstCase}
\underline{v}_I(\mu) & = \inf_{\pi_1,\pi_1^*,\pi_0,\pi_0^*,q \in [0,1]} q(2\pi_1^*-1) \mbox{ s.t. }\\
& \quad \quad \quad \quad 
\begin{array}{cc}
q\pi_1^* + (1-q)\pi_0^*=p & \mbox{(BP-1)} \\
q\pi_1 + (1-q)\pi_0=p & \mbox{(BP-2)} \\
\pi_1 \geq 1/2 > \pi_0 & \mbox{(IC)} \\
q = \mu(E) \mbox{ for some } E \subseteq \mathcal{X}_1 & \mbox{(M)}
\end{array}
\end{align*}
and the best-case payoff is
\begin{equation*}
\overline{v}_I(\mu) = \sup_{\pi_1,\pi_1^*,\pi_0,\pi_0^*,q} q(2\pi_1^*-1) \mbox{ s.t. (BP-1),(BP-2),(IC),(M)}.
\end{equation*}
\end{lemma}

\begin{proof}
Fix any $X_1 \sim \mu$. We will show that there exist admissible $P,P^* \in \Delta(\mathcal{X}_1 \times \mathcal{Y})$ yielding  expected payoff  $v$ if and only if $v  = q(2\pi_1^*-1)$ for some $q,\pi_1,\pi_1^*,\pi_0,\pi_0^* \in [0,1]$ satisfying (BP-1),(BP-2),(IC), and (M).  

For the only-if direction, fix any admissible $P,P^* \in \Delta(\mathcal{X}_1 \times \mathcal{Y})$. Then the designer takes action $a=1$ on the event
\[A_1 := \left\{x_1 \in \mathcal{X}_1 : P(Y=1 \mid X_1 = x_1) \geq \frac12\right\}\]
and $a=0$ on the complementary event $A_0 := \mathcal{X}_1 \backslash A_1$. Define $q=P^*(A_1)$ to be the ex-ante probability with which treatment is recommended, and
 \begin{align*}
     \pi_1^* & = P^*(Y=1 \mid A_1) \quad \quad \pi_1 = P(Y=1 \mid A_1)  \\
      \pi_0^* & = P^*(Y=1 \mid A_0) \quad \quad \pi_0 = P(Y=1 \mid A_0)
 \end{align*}
  to be the conditional probabilities of need of treatment when treatment is or isn't recommended, respectively according to the distributions $P^*$ and $P$. 
  
The designer's expected payoff is
\begin{align*}
 \mathbb{E}_{P^*}[u(Y,a(X,P))] & = P^*(A_1) P^*(Y=1 \mid A_1 ) - P^*(A_1) P^*(Y=0 \mid A_1) \\
 & = P^*(A_1) ( 2P^*(Y=1 \mid A_1) - 1) = q(2\pi_1^*-1)
\end{align*}
so it remains to show that $q,\pi_1^*,\pi_0^*,\pi_1,\pi_0$ satisfy (BP-1), (BP-2), (IC), and (M).

By Part (1) of admissibility, $P^*(A_1) = \mu(A_1)$, so (M) follows immediately by definition of $A_1$ as a subset of $\mathcal{X}_1$. Moreover the definitions of $q$ and $\pi_0,\pi_1$ imply (IC). Finally (BP-1) and (BP-2) follow from Part (2) of admissibility, which requires  $P(Y=1)=P^*(Y=1)=p$.

To show the other direction, suppose there are $q,\pi_1^*,\pi_0^*,\pi_1,\pi_0$ satisfying (BP-1), (BP-2), (IC), and (M), where $v=q(2\pi_1^*-1)$. By (M) there is an event $E \subseteq \mathcal{X}_1$ such that $q=\mu(E)$. Now define distributions $P^*,P \in \Delta(\mathcal{X}_1 \times \mathcal{Y})$ such that their marginal distribution on $\mathcal{X}_1$ is $\mu$, and
\[P^*(Y=1 \mid X_1 =x_1 ) = \left\{\begin{array}{cc}
\pi^*_1 & \mbox{ if } x_1 \in E \\
\pi^*_0 & \mbox{ if } x_1 \notin E 
\end{array}\right.\]
\[P(Y=1 \mid X_1 =x_1 ) = \left\{\begin{array}{cc}
\pi_1 & \mbox{ if } x_1 \in E \\
\pi_0 & \mbox{ if } x_1 \notin E 
\end{array}\right.\]
Since by assumption, $\mbox{marg}_{\mathcal{X}_1} P = \mbox{marg}_{\mathcal{X}_1} P^* = \mu$, both $P$ and $P^*$ satisfy Part (1) of admissibility. Moreover (BP-1) and (BP-2) ensure Part (2). Finally, since $\pi_1 \geq 1/2 > \pi_0$ by (IC), $E$ is precisely the event on which treatment is recommended given $P$. So the designer's expected payoff is
\[P^*(E)P^*(Y=1 \mid E) - P^*(E) P^*(Y=0 \mid E) = q(2\pi_1^* -1) = v\]
as desired.
\end{proof}
\section{Proof of Lemmas \ref{lemm:p} and \ref{lemm:pImplement} } \label{appendix A}

By assumption, $\mathcal{X}_0$ is  a singleton. To ease notation we will write $p\equiv p_{x_0}$ for the prior probability of $Y=1$, and $\mathcal{P}(p,\mu)$ for the set of joint distributions $P$ that are admissible.

\subsection{Plan of the proof}

We will solve for $\underline{v}_I(\mu)$ and $\overline{v}_I(\mu)$ as given in Lemma \ref{lemm:Restate}. 
These were defined as an  infimum and a supremum over the variables $\pi_1,\pi_1^*,\pi_0,\pi_0^*,$ and $q \in [0,1]$ given the constraints  (BP-1),(BP-2),(IC), and (M). 
Say that $q$ is feasible if there exist $\pi_1^*, \pi_0^*, \pi_1,\pi_0$ that satisfy these constraints. Since we can always take $\pi_1^*=\pi_1$ and $\pi_0^*=\pi_0$, the constraint (BP-1) is redundant here, and the feasible set of treatment probabilities is \begin{equation} \label{eq:feasibleq}
\mathcal{Q}(p,\mu) = \{q \in [0,1]: \exists (\pi_1,\pi_0) \mbox{ s.t. (BP-2),(IC),(M)}\}.
\end{equation} 
Section \ref{sec:Feasibleq}  characterizes this set.

Section \ref{sec:phi} fixes an arbitrary feasible treatment probability $q \in \mathcal{Q}(p,\mu)$. Since the remaining payoff-relevant variable $\pi_1^*$ is constrained only by (BP-1), we provide closed-form solutions for
\begin{align*}
   \phi_W(q) \equiv \min_{\pi_1^*,\pi_0^* \in [0,1]} q(2\pi_1^* -1) \mbox{ s.t. (BP-1)}\\
   \phi_B(q) \equiv \max_{\pi_1^*,\pi_0^* \in [0,1]} q(2\pi_1^*  - 1) \mbox{ s.t. (BP-1)}
\end{align*}

Putting these parts together, the worst-case and best-case problems in Lemma \ref{lemm:Restate} can be restated as
\begin{align*}
\underline{v}_I(\mu) &= \inf_{q \in \mathcal{Q}(p,\mu)} \phi_W(q) \\
\overline{v}_I(\mu) &= \sup_{q \in \mathcal{Q}(p,\mu)} \phi_B(q)    
\end{align*}

Varying over $\mu$ thus produces a set of possible values for $(\underline{v}_I(\mu),\overline{v}_I(\mu))$, and we finally solve for the $\mu$ that implement $(\underline{v}_I(\mu),\overline{v}_I(\mu))$ that are undominated in the set
\[\mbox{conv}\{(\underline{v}_I(\mu),\overline{v}_I(\mu)) : \mu \in M\}\]
where $M$ is the set of all finitely supported measures. Section \ref{sec:2point} demonstrates that it is sufficient to consider distributions $\mu$ with a two-point support, and Section \ref{sec:frontier} uses this reduction to characterize the frontier. 

\subsubsection{The feasible values of $q$} \label{sec:Feasibleq}

We first solve for $\mathcal{Q}(p,\mu)$ as defined in (\ref{eq:feasibleq}), i.e., the values of $q$ satisfying (BP-2), (IC), and (M). Observe that we can decompose this set as $\mathcal{Q}(p,\mu) = BP(p) \cap \mathcal{E}(\mu)$ where
\[BP(p) = \{ q\in [0,1] : \exists (\pi_1,\pi_0) \mbox{ s.t. (BP-2) and (IC)}\}\]
are the probabilities with which treatment can  be recommended in an incentive-compatible way, and
\[\mathcal{E}(\mu) = \{q \in [0,1] : q \mbox{ satisfies (M)}\} = \left\{ \mu(E) : E \subseteq \mathcal{X}_1\right\}\]
is the set of all attainable event probabilities under $\mu$. 

The following lemma characterizes $BP(p)$.

\begin{lemma} \label{lemm:BP} For any $p\in [0,1]$, 
\[BP(p) = \left\{\begin{array}{cc}
 [0,2p] & \mbox{ if } 0\leq p < 1/2  \\
(0,2p] & \mbox{ if } p=1/2 \\
 (2p-1,1] & \mbox{ if }  1/2 < p \leq 1
\end{array}\right.\]
\end{lemma}
\noindent This set is depicted in Figure \ref{fig:Qp}.

\begin{figure}[h]
    \includegraphics[scale=0.22]{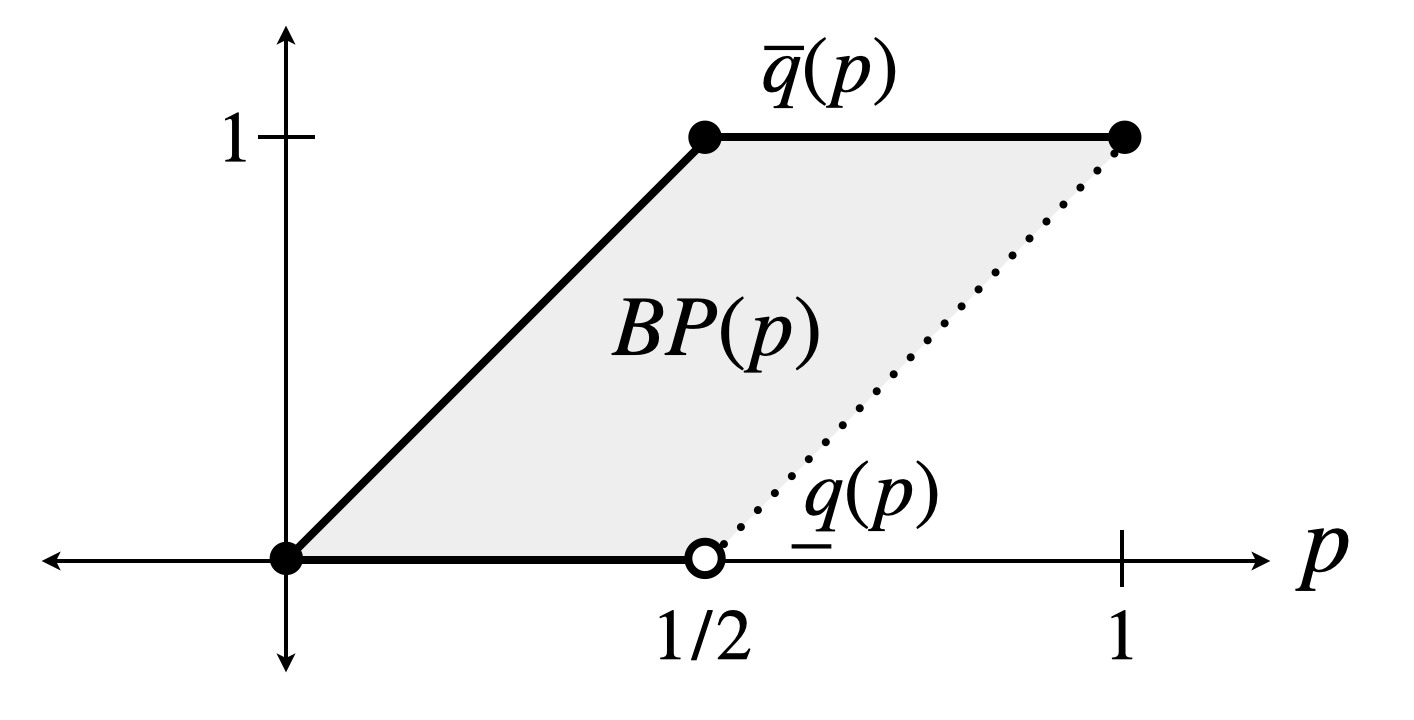}
    \caption{Illustration of $BP(p)$.} \label{fig:Qp}
\end{figure}

\begin{proof}
Rewriting (BP-2) we have
\[q = \frac{p-\pi_0}{\pi_1-\pi_0}\]
where $0\leq \pi_0\leq p \leq \pi_1 \leq 1$. Define
\[\underline{q}(p) = \inf_{\pi_0,\pi_1 \in [0,1]} \frac{p-\pi_0}{\pi_1-\pi_0} \quad \mbox{s.t. } \pi_0\leq p \leq \pi_1  \mbox{ and } \pi_0 < \frac12 \leq \pi_1 \]
\[\overline{q}(p) = \sup_{\pi_0,\pi_1 \in [0,1]} \frac{p-\pi_0}{\pi_1-\pi_0} \quad \mbox{s.t. } \pi_0\leq p \leq \pi_1 \mbox{ and } \pi_0 < \frac12 \leq \pi_1\]
The expression $\frac{p-\pi_0}{\pi_1-\pi_0}$ is decreasing in both $\pi_1$ and $\pi_0$. Thus we have
\[\underline{q}(p) = \left\{ \begin{array}{cc}
0 & \mbox{ if } p \in [0,1/2] \\
2p-1 & \mbox{ if } p \in (1/2,1]
\end{array}\right.\]
as attained exactly by $\pi_1=1$ and $\pi_0=p$ on $p \in [0,1/2)$, as the limit when $\pi_1 \uparrow \frac12$ and $\pi_0 \uparrow \frac12$ at $p=1/2$, and as the limit when $\pi_1=1$ and $\pi_0 \uparrow 1/2$ on $p \in (1/2,1]$. Moreover
\[\overline{q}(p) = \left\{ \begin{array}{cc}
2p & \mbox{ if } p \in [0,1/2] \\
1 & \mbox{ if } p \in (1/2,1]
\end{array}\right.\]
as attained by $\pi_1=1/2$ and $\pi_0=0$ on $p \in [0,1/2]$, and $\pi_1=p$ and $\pi_0=0$ on $p \in (1/2,1]$.
\end{proof}

\subsubsection{Solving the problem given a fixed $q$} \label{sec:phi}

Fix an arbitrary $q$ chosen by the AI agent. A benevolent Nature  maximizes the designer's expected payoff by solving
\begin{align*}
   \phi_B(q) = \max_{\pi_1^*,\pi_0^* \in [0,1]} q(2\pi_1^* & - 1) \\
   \mbox{s.t. } & \,\, q\pi_1^* + (1-q) \pi_0^* = p
\end{align*}
Since $q(2\pi_1^*-1)$ is increasing in $\pi_1^*$, this problem corresponds to maximizing $\pi_1^*$ subject to the constraint that there exists a $\pi_0^*$ such that $q\pi_1^* + (1-q) \pi_0^* = p$.
  There are two relevant cases.  If $p>q$, then $\pi_1^*=1$ is feasible (setting $\pi_0^*=\frac{p-q}{1-q}\in [0,1]$), in which case $q(2\pi_1^*-1)=q$. If instead $p \leq q$ then the maximum feasible value is $\pi_1^* = p/q$ (now setting $\pi_0^*=0$), implying $q(2\pi_1^*-1)=2p-q$. So the solution  is
\[\phi_B(q) = \left\{\begin{array}{cc}
q & \mbox{ if } q < \hat{q}(p) \\
2p-q & \mbox{ if } q \geq \hat{q}(p) \end{array}\right.\]
where $\hat{q}(p)=p.$

The lowest possible expected payoffs for the designer are
\begin{align*}
   \phi_W(q) = \min_{\pi_1^*,\pi_0^* \in [0,1]} q(2\pi_1^* & - 1) \\
   \mbox{s.t. } & \,\, q\pi_1^* + (1-q) \pi_0^* = p
\end{align*}
Here Nature wants to minimize $\pi_1^*$ subject to feasibility. If $\frac{p}{1-q}\leq 1$, then $\pi_1^*=0$ is feasible (setting $\pi_0^*=\frac{p}{1-q}\in [0,1]$), in which case $q(2\pi_1^*-1)=-q$. If instead $\frac{p}{1-q}>1$ then the minimum feasible value is $\pi_1^* = 1 - \frac{1-p}{q}$ (now setting $\pi_0^*=1$), implying $q(2\pi_1^*-1)=q - 2(1-p)$. So the solution to the above program is
\[\phi_W(q) = \left\{\begin{array}{cc}
-q & \mbox{ if } q\leq 1-\hat{q}(p) \\
q-2(1-p) & \mbox{ if } q> 1-\hat{q}(p) \end{array}\right.\]

Figure \ref{fig:phi} depicts the set of values of $(\phi_W(q),\phi_B(q))$ as $q$ ranges from 0 to 1. 

\begin{figure}[h]
    \includegraphics[scale=0.3]{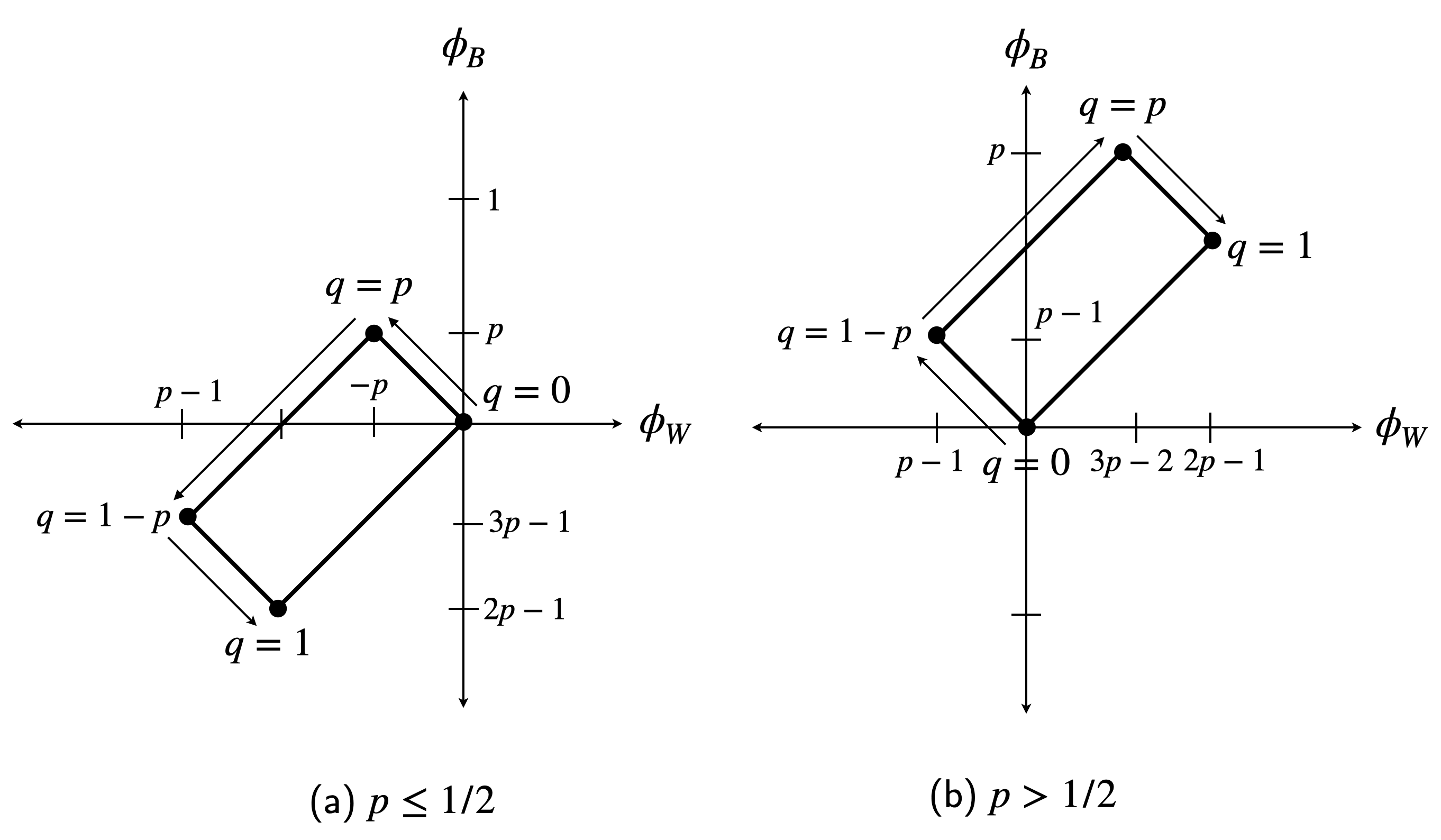}
    \caption{Illustration of $(\phi_W(q),\phi_B(q))$ as $q$ ranges from 0 to 1 for $p\leq 1/2$ (left) and $p>1/2$ (right).} \label{fig:phi}
\end{figure}

\subsubsection{The agent's choice of $q$} 
Given Nature's best reply (as characterized in Section \ref{sec:phi}), the AI agent optimizes over $q \in \mathcal{Q}(p,\mu)$ (characterized in Section \ref{sec:Feasibleq}). That is, the worst-case problem is equivalent to
\[\underline{v}_I(\mu) = \inf_{q \in \mathcal{Q}(p,\mu)} \phi_W(q)\]
and the best-case problem is equivalent to
\[\overline{v}_I(\mu) = \sup_{q \in \mathcal{Q}(p,\mu)} \phi_B(q)\]

It remains to determine the frontier of 
\[\mbox{conv}\{(\underline{v}_I(\mu),\overline{v}_I(\mu)) : \mu \in M\}\]
where $M$ is the set of all finitely supported measures. 

\subsubsection{Sufficiency of two-point supports} \label{sec:2point}

We first argue that to identify the set of undominated payoff vectors, it is sufficient to consider $\mu$ supported on two points. That is, the Pareto frontier is a subset of
\[\{(\underline{v}_I(\mu_q),\overline{v}_I(\mu_q)) : q\in [0,1]\}\]
where $\mu_q=(q,1-q)$ is the distribution of $X_1 \sim \mbox{Bernoulli}(q)$. Intuitively, once we've identified the value of $q$ that an aligned AI would pick, the other feasible values of $q$ do not improve the designer's best-case payoff and weakly reduce their worst-case payoff, so it's better to replace $\mu$ with the minimal distribution that permits the aligned AI's choice of $q$.\footnote{This intuition extends to settings where there are more than two possible actions, but not to settings with a third type of AI that is only partially aligned with the designer.}

In more detail, we will argue that every $(\underline{v}_I(\mu),\overline{v}_I(\mu))$ is weakly dominated by $(\underline{v}_I(\mu_q),\overline{v}_I(\mu_q))$ for some $q$. Consider an arbitrary $\mu$ and let
\[q^* = \argmax_{q \in \mathcal{Q}(p,\mu)} \phi_B(q).\]
be the value of $q$ that attains the best-case payoff.\footnote{Since $\mu$ has finite support, $\mathcal{E}(\mu)$ is a finite set, so the max is attained.}

Now define $\mu_{q^*} = (q^*,1-q^*)$. By assumption, $q^*$ is an attainable event probability under $\mu$, which implies that $1-q^*$ must be as well. So $\mathcal{E}(\mu) \supseteq \{0,q^*,1-q^*,1\} = \mathcal{E}(\mu_{q^*})$, i.e., the set of attainable event probabilities under $\mu_{q^*}$ is a subset of those attainable under $\mu$. Thus
\[\mathcal{Q}(p,\mu_{q^*}) = BP(p) \cap \mathcal{E}(\mu_{q^*}) \subseteq BP(p) \cap \mathcal{E}(\mu) = \mathcal{Q}(p,\mu),\]
which implies
\[\underline{u} \equiv \min_{q \in \mathcal{Q}(p,\mu)} \phi_W(q) \leq \min_{q \in \mathcal{Q}(p,\mu_{q^*})} \phi_W(q)  \equiv \underline{u}' \]
and
\[\overline{u} \equiv \max_{q \in \mathcal{Q}(p,\mu)} \phi_B(q) \geq \max_{q \in \mathcal{Q}(p,\mu_{q^*})} \phi_B(q)  \equiv \overline{u}'\]
But the latter must hold with equality because $q^*$ attains the max on the LHS (by definition of $q^*$) and $q^* \in \mathcal{Q}(p,\mu_{q^*})$ (by construction of $\mu_{q^*}$). So $(\underline{u}',\overline{u}') \geq (\underline{u},\overline{u})$ as desired.

\subsubsection{Undominated payoff vectors} \label{sec:frontier}

Applying the reduction from the previous section, the frontier is a subset of
\[\mbox{conv}\{(\underline{v}_I(\mu_q),\overline{v}_I(\mu_q)) : q \in [0,1/2]\}\]
where it is without loss to restrict $q<1/2$ since $\mathcal{E}(\mu_q) = \mathcal{E}(\mu_{1-q})$ (i.e., these distributions imply the same atainable event probabilities). 

Suppose that $p<1/2$, implying $1-\hat{q}(p)=1-p<p=\hat{q}(p)$. Besides the thresholds $\hat{q}(p)$ and $1-\hat{q}(p)$, a third key value of $q$ is  
$\overline{q} = 1-2p.$ The importance of this threshold is seen as follows. By Lemma \ref{lemm:BP},  the set of incentive-compatible values of $q$ for $p<1/2$ is $[0,2p]$. For the two-support distribution $\mu_q$ (corresponding to $X_1 \sim \mbox{Bernoulli}(q)$), the attainable event probabilities are $\{0,q,1-q,1\}$. If $q<\overline{q}$, then among the attainable event probabilities only 0 and $q$ are incentive compatible, so $\mathcal{Q}(p,\mu_q)=\{0,q\}$. But if $q>\overline{q}$, then  $1-q$ is also incentive compatible, so $\mathcal{Q}(p,\mu_q)=\{0,q,1-q\}$.

With this in mind, partition the values of $q \in [0,1/2]$ into three sets (some of which may be empty):
\[Q_1 = [0,\min\{\hat{q}(p),\overline{q}(p)\}) \quad \quad Q_2 = [\overline{q}(p),\hat{q}(p)] \quad \quad Q_3 = (\hat{q}(p),1/2).\]
For any $\mu_q$ with $q \in Q_1$, the feasible set is $\mathcal{Q}(p,\mu_q) = \{0,q\}$, and 
\begin{align*}
   (\underline{v}_I(\mu_q),\overline{v}_I(\mu_q)) & =  \left(\min_{q' \in \{0,q\}} \phi_W(q'),\max_{q' \in \{0,q\}} \phi_B(q') \right) \\ & = (\phi_W(q),\phi_B(q)) = (-q,q).
\end{align*}
For any $\mu_q$ with $q \in Q_2$, the feasible set is $\mathcal{Q}(p,\mu_q) = \{0,q,1-q\}$ and
\begin{align*}
(\underline{v}_I(\mu_q),\overline{v}_I(\mu_q)) & = \left(\min_{q' \in \{0,q,1-q\})} \phi_W(q'),\max_{q' \in \{0,q,1-q\}} \phi_B(q') \right) \\ & = (\phi_W(1-q),\phi_B(q))  = (2p-q-1,q).
\end{align*}
And finally for any $\mu_q$ with $q \in Q_3$, the feasible set is $\mathcal{Q}(p,\mu_q) = \{0,q,1-q\}$ and
\begin{align*}
    (\underline{v}_I(\mu_q),\overline{v}_I(p,\mu_q)) & = \left(\min_{q' \in \{0,q,1-q\}} \phi_W(q'),\max_{q' \in \{0,q,1-q\}} \phi_B(q') \right) \\ & = (\phi_W(1-q),\phi_B(q)) = (q-1,2p-q).
\end{align*}
Define
 \begin{align*}
   S_1 &= \{(-q,q): q \in Q_1\}, \\
   S_2 &= \{(2p-q-1,q) : q\in Q_2\},\text{ and } \\
   S_3 &= \{(q-1,2p-q) : q \in Q_3.\}
 \end{align*}
 Then the closure of the Pareto frontier are those points that are undominated in
 \[\mbox{cl}(\mbox{conv}(S_1 \cup S_2 \cup S_3)) = \mbox{conv}(\mbox{cl}(S_1) \cup \mbox{cl}(S_2) \cup \mbox{cl}(S_3)),\]
 where equality follows because $S_1,S_2,S_3 \subseteq \mathbb{R}^2$.

 When $\hat{q}(p)<\overline{q}(p)$ (or equivalently, $p <1/3$), then $Q_1 = [0,\hat{q}(p)]$ and $Q_2 $ is empty. The set $\mbox{cl}(S_1)$ is the line segment connecting $(-p,p)$ to the origin, while $\mbox{cl}(S_3)$ is the line segment connecting $(p-1,p)$ to $(-1/2,2-p-1/2)$. Every point on $\mbox{cl}(S_3)$ is dominated by a point on $\mbox{cl}(S_1)$, while the points on $\mbox{cl}(S_1)$ are unranked. So the closure of the frontier is simply $\mbox{cl}(S_1).$

When $\hat{q}(p) \geq \overline{q}(p)$ (equivalently, $p \in [1/3,1/2)$),  then $Q_1= [0,1-2p)$ and $Q_2 = [1-2p,p]$. The set $\mbox{cl}(S_1)$ is the line segment connecting $(2p-1,1-2p)$ to the origin, while both $\mbox{cl}(S_2)$ and $\mbox{cl}(S_3)$ are line segments with left endpoint $(p-1,p)$ and slope $-1$. This line lies everywhere below the line connecting $(p-1,p)$ and $(2p-1,1-2p)$, whose slope $(1-3p)/p$ is strictly larger than $-1$ for all $p< 1/2$.  So all of those points are dominated by a point that can be achieved by randomizing between $(p-1,p)$ and $(2p-1+\varepsilon,1-2p-\varepsilon)$ for $\varepsilon$ sufficiently small  (see Figure \ref{fig:dominated}). 

\begin{figure}[h]
    \includegraphics[scale=0.2]{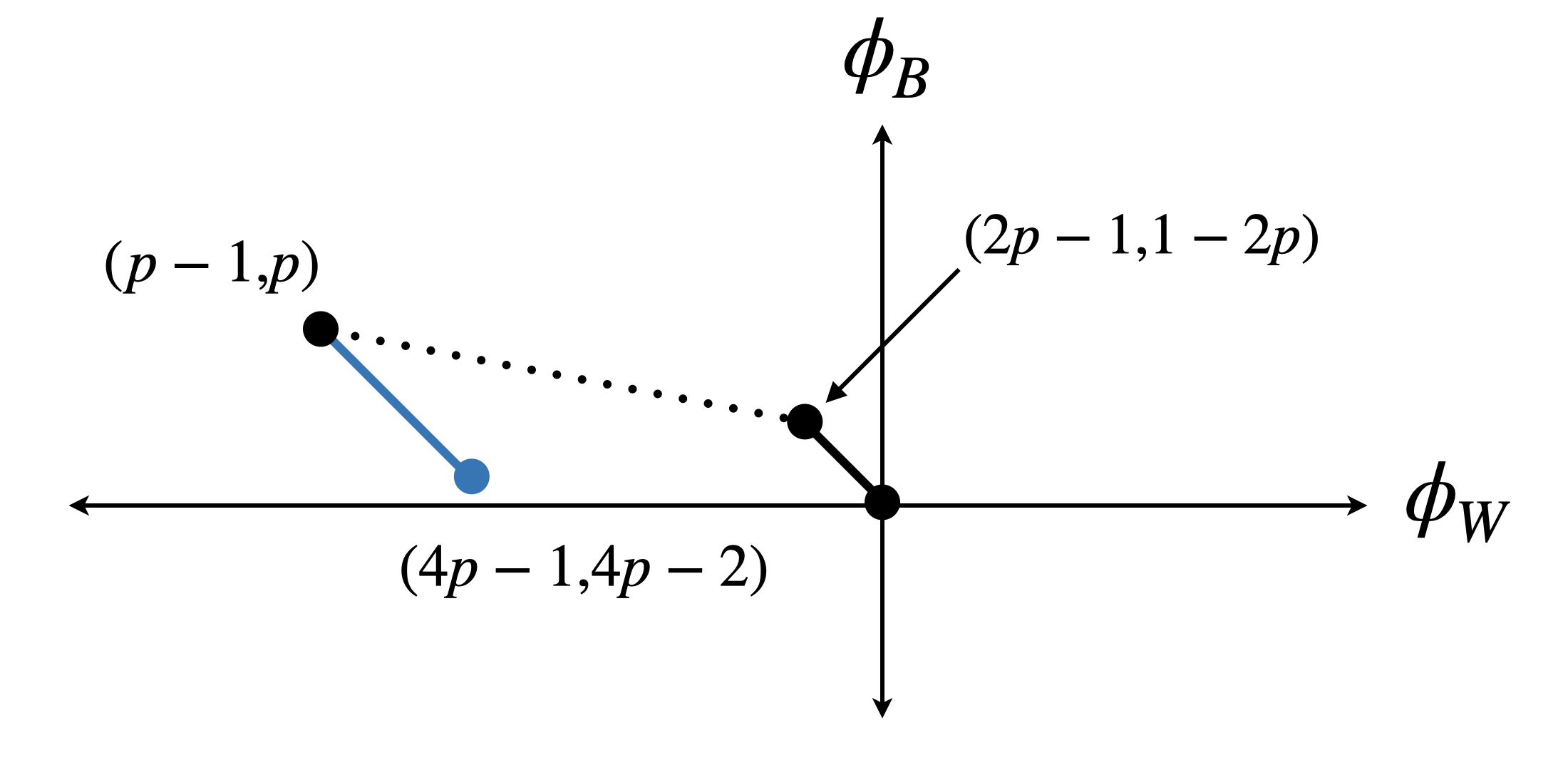}
    \caption{\small{The points on the line segment connecting $(p-1,p)$ to $(4p-1,4p-2)$ are dominated by points achieved by randomizing between $(p-1,p)$ and $(2p-1,1-2p)$ Thus those points (blue) are not on the frontier.}} \label{fig:dominated}
\end{figure}

The points on the union of the line segment connecting $(p-1,p)$ to $(2p-1,1-2p)$, and the line segment connecting $(2p-1,1-2p)$ to $(0,0)$ cannot be Pareto ranked, and so these constitute the closure of the full frontier for the case $p \in [1/3,1/2)$.

Finally if $p=1/2$ then $Q_1=\varnothing$ while $Q_2=[0,1/2]$. The set $\mbox{cl}(S_2)$ is the line segment from (0,0) to $(-1/2,1/2)$ and this constitutes the closure of the frontier.

The argument for $p>1/2$ proceeds nearly identically and is thus omitted.

\section{Proofs of Theorems \ref{thm:ParetoGeneral} and \ref{thm:supportGeneral}} \label{sec:Minkowski}

Now consider $I =(\mathcal{X}_0, \mu_0, (p_{x_0})_{x_0 \in \mathcal{X}_0})$ where $\mathcal{X}_0$ is not necessarily a singleton. For each $x_0$, let $I_{x_0}=(\{x_0\},\delta_{x_0},p_{x_0})$ correspond to prior information about subgroup $x_0$ only. Then choose a pair $(\underline{v}_{x_0}, \overline{v}_{x_0}) \in \mathcal{F}(I_{x_0},\mu)$ from each subgroup $x_0$'s frontier, corresponding to worst-case and best-case payoffs for patients in that subgroup. Because payoffs add linearly across $x_0$ and the prior weights $\mu_{x_0} := \mu_0(x_0)$ are fixed, all payoffs in 
\begin{equation} \label{eq:C}
C = \left\{ \sum_{x_0 \in \mathcal{X}_0} \mu_{x_0} (\underline{v}_{x_0}, \overline{v}_{x_0}) : (\underline{v}_{x_0}, \overline{v}_{x_0})\in \mathcal{F}(I_{x_0},\mu)\right\}
\end{equation}
are in the closure of the feasible set. This is the  weighted Minkowski sum of the sets $\{\mathcal{F}(I_{x_0},\mu) : x_0 \in \mathcal{X}_0\}$. The set of undominated points in $C$ is the efficient frontier $\mathcal{F}(I)$. That is, 
\[\mathcal{F}(I) = \{(\underline{v},\overline{v}) \in C : \nexists (\underline{v}',\overline{v}') \in C \mbox{ s.t. } (\underline{v}',\overline{v}') \succ (\underline{v},\overline{v})\}\]

It  remains to show that the upper frontier of $C$ follows the characterization given in Proposition  \ref{thm:Pareto}. For any preference parameter $\alpha$
define the linear functional
\[\Lambda_\alpha(\underline{v},\overline{v}) = \alpha \underline{v} + (1-\alpha) \overline{v}.\]
By linearity of $\Lambda_\alpha$ and the form of the Minkowski sum in (\ref{eq:C}), maximizing $\Lambda_\alpha$ over the elements of $C$ is equivalent to maximizing $\Lambda_\alpha$ covariate by covariate. That is, if for each $x_0$ we find a maximizer $(\underline{v}_{x_0}, \overline{v}_{x_0}) \in \mathcal{F}(p_{x_0},\mu)$, then $\sum_{x_0 \in \mathcal{X}_0} \mu_{x_0}(\underline{v}_{x_0}, \overline{v}_{x_0})$ is a maximizer over the elements of $C$, and moreover every maximizer over $C$ can be decomposed in this way.

So now fix an arbitrary $x_0$. Given Lemma \ref{lemm:p}, a maximizer of $\Lambda_\alpha$ must be one of the (at most) three extreme points: $T_{x_0}$, $H_{x_0}$, or $D_{x_0}$. Since a designer with preference parameter $\alpha$ has indifference curves that are lines with slope $-\frac{\alpha}{1-\alpha}$, we can determine which  point is optimal by comparing $-\frac{\alpha}{1-\alpha}$ to the slopes of $\overline{T_{x_0} H_{x_0}}$ and $\overline{H_{x_0} D_{x_0}}$. Recall that $\sigma(x_0) \in (-1,0]$ (as defined in Section \ref{sec:GeneralResults}) is the slope of the line $\overline{T_{x_0} H_{x_0}}$, while the slope of $\overline{H_{x_0}D_{x_0}}$ is $-1$ for every $px_0$. Thus one optimal point is
\[(\underline{v}(x_0), \overline{v}(x_0)) = \left\{\begin{array}{cc}
T_{x_0} & \mbox{ if } -\frac{\alpha}{1-\alpha} > \sigma(p_{x_0}) \\
H_{x_0} & \mbox{ if } \sigma(p_{x_0}) \geq -\frac{\alpha}{1-\alpha} \geq -1 \\
D_{x_0} & \mbox{ if } -1 \geq -\frac{\alpha}{1-\alpha}
\end{array}\right.\]

For each $j=1, \dots, k$ define $\alpha_j$ to satisfy
\[-\frac{\alpha_j}{1-\alpha_j} = \sigma_j.\]
Note that $\alpha_j$ are monotonically increasing in $j$ and satisfy $0 \leq \alpha_j \leq 1/2$, while $\sigma_j$ are monotonically decreasing in $j$ and satisfy $0 \geq \sigma_j \geq -1$.   So for $\alpha \in [0,\alpha_1]$, $-\alpha/(1-\alpha) > \sigma_j$ for all $j$, meaning  the optimal choice at each $x_0$ is $T_{x_0}$, and the optimal choice in $C$ is the trust point $\mathsf{T} = \sum_{x_0 \in \mathcal{X}_0} \mu_{x_0} T_{x_0}$. For each $\alpha \in [\alpha_j,\alpha_{j+1}]$, the hedge point $H_{x_0}$ is optimal for $x_0 \in x_0^{(1)}, \dots, x_0^{(j)}$, while $T_{x_0}$ is optimal for all remaining $x_0$. And finally for $\alpha \in [1/2,1]$, the point $D_{x_0}$ is optimal  for every $x_0$. Thus as $\alpha$ ranges from 0 to 1, the maximizers trace the line
\begin{equation} \label{eq:Chain}
\overline{\mathsf{T}\mathsf{H}_1} \cup \overline{\mathsf{H}_1\mathsf{H}_2} \cup \dots \cup \overline{\mathsf{H}_{k-1}\mathsf{H}_k} \cup \overline{\mathsf{H}_k\mathsf{D}}
\end{equation}
with slopes $\sigma_1, \dots, \sigma_k, -1$. So the union of these line segments must constitute a part of the frontier.

We finally argue that no other points can be a part of the frontier. Weighted Minkowski sums preserve the convexity and compactness of the individual frontiers $F(I_{x_0},\mu)$, so  the set $C$ is compact and convex. Therefore  every undominated point $v\in C$ is \emph{supported}, i.e., there exists an $\alpha$ such that $v=\argmax_{w \in C} \Lambda_\alpha(w)$. We have already shown that  these maximizers precisely trace (\ref{eq:Chain}). Hence any point outside of this union cannot be supported, and therefore cannot belong to the efficient frontier.

\section{Tie-Breaking} \label{app:TieBreak}

In the main text we supposed that the designer breaks ties in favor of treatment. We now show that this is without loss in the following sense. Suppose that the decision rule in (\ref{eq:aStar}) is replaced with a mixed strategy $\alpha: \mathcal{X}_0 \times \mathcal{X}_1 \rightarrow [0,1]$ mapping covariates into a probability of action $a=1$, where
\begin{equation}
\alpha(x_0,x_1) = \left\{\begin{array}{cc}
1 & \mbox{if } \mathbb{E}_{P}(Y\mid (X_0,X_1)=(x_0,x_1)) > \frac{1}{2} \\
\rho & \mbox{if } \mathbb{E}_{P}(Y\mid (X_0,X_1)=(x_0,x_1)) = \frac{1}{2} \\
0 & \mbox{otherwise}
\end{array}\right.
\end{equation}
and the randomization parameter is $\rho \in [0,1]$. Let $(\underline{v}_\rho(\mu),\overline{v}_\rho(\mu))$ denote the designer's worst-case and best-case payoffs given this decision rule and any $X_1 \sim \mu$. 

\begin{claim} For every $\rho<1$ and $\mu$, there is another $\mu'$ such that $\underline{v}_1(\mu') \geq \underline{v}_\rho(\mu)$ and $\overline{v}_1(\mu') \geq \overline{v}_\rho(\mu)$, i.e. $(\underline{v}_1(\mu'),\overline{v}_1(\mu'))$ \emph{weakly dominates} $(\underline{v}_\rho(\mu),\overline{v}_\rho(\mu))$.
\end{claim}

Thus the designer can optimally choose $\rho=1$. The proof shows that if the designer randomizes, the best-case and worst-case payoffs are a convex combination of the distrust payoffs $D$ and the $\rho=1$ payoffs. Since the frontier is the upper boundary of a convex set, these convex combinations are  weakly dominated relative to the original frontier. Intuitively, since the designer's randomization is independent of the AI's choices, neither type of AI can exploit this randomization to effect its goals. So randomization is similar to blunting the AI's capabilities. 

\begin{proof} Suppose first that $p<1/2$. Then necessarily $\pi_0 <1/2$, so indifference can only occur if $\pi_1 = 1/2$. The designer's payoff in this case is $\rho q (2\pi_1^*-1)$, and it is straightforward to see that the values of $(q,\pi_1^*)$ that maximized (or minimized) the original objective $q(2\pi_1^*-1)$ likewise maximize (or minimize) this new objective. So 
\[(\underline{v}_\rho(\mu),\overline{v}_\rho(\mu)) = \rho (\underline{v}_1(\mu),\overline{v}_1(\mu))\]
i.e., $(\underline{v}_\rho(\mu),\overline{v}_\rho(\mu))$ lies on the line segment connecting $(\underline{v}_1(\mu),\overline{v}_1(\mu))$ and the origin. But the frontier is convex, and for $p<1/2$ the origin is an extreme point of the $\rho=1$ frontier. So $(\underline{v}_\rho(\mu),\overline{v}_\rho(\mu))$ either lies on the $\rho=1$ frontier or is dominated by some point on the $\rho=1$ frontier.  

Now suppose $p>1/2$. Then necessarily $\pi_1 >1/2$, so indifference can only occur if $\pi_0=1/2$. The designer's payoff in this case is 
\[q(2\pi_1^*-1) + \rho(1-q)(2\pi_0^*-1)\]
where $\pi_1^*$ and $\pi_0^*$ satisfy $q\pi_1^* +(1-q)\pi_0^*=p$. Rewriting the constraint we have
\[\pi_0^* = \frac{p - q \pi_1^*}{1-q}\]
and plugging this into the objective allows us to rewrite it as
\[(1-\rho)q(2\pi_1^*-1) + \rho(2p-1).\]
Since $1-\rho>0$, this objective is maximized (and minimized) by the same values of $(q,\pi_1^*)$ as the original objective, and we can write
\[(\underline{v}_\rho(\mu),\overline{v}_\rho(\mu)) = (1-\rho) (\underline{v}_1(\mu),\overline{v}_1(\mu)) + \rho (2p-1,2p-1)\]
i.e., $(\underline{v}_\rho(\mu),\overline{v}_\rho(\mu))$ lies on the line segment connecting $(\underline{v}_1(\mu),\overline{v}_1(\mu))$ and $(2p-1,2p-1)$. But for $p>1/2$, $(2p-1,2p-1)$ is an extreme point of the $\rho=1$ frontier. So by the same argument as above,  $(\underline{v}_\rho(\mu),\overline{v}_\rho(\mu))$ either lies on the $\rho=1$ frontier or is dominated by some point on it, as desired.  

Finally the case $p=1/2$ is trivial since  a tie can only occur if $\pi_1^*=\pi_0^*=1/2$, in which case the frontier collapses to the point $(0,0)$ regardless of $\rho$.
\end{proof}

\section{Proofs of Lemmas \ref{lemm:pConstrained} and \ref{lemm:pImplementConstrained}}\label{sec:constrained}

The following proof is very similar  to the proof of Lemmas \ref{lemm:p} and \ref{lemm:pConstrained}, so  we focus on the parts of the argument that are new.

\subsection{The feasible values of $q$}

We generalize $BP(p)$ to 
\[BP_\tau(p) = \{ q\in [0,1] : \exists \pi_1,\pi_0 \in [\ltau,\utau]  \mbox{ s.t. (BP-2) and (IC)}\}\]

\begin{lemma}  \label{lemm:BPtau} For any $\tau = (\ltau,\utau)$ and $p\in [\ltau,\utau]$, 
\[BP_\tau (p) = \left\{\begin{array}{cc}
 \left[0,\frac{2(p-\ltau)}{1-2\ltau}\right] & \mbox{ if } 0\leq p < 1/2  \\
\left(0,\frac{2(p-\ltau)}{1-2\ltau}\right] & \mbox{ if } p=1/2 \\
 \left(\frac{2p-1}{2\utau - 1},1\right] & \mbox{ if }  1/2 < p \leq 1
\end{array}\right.\]
\end{lemma}

\begin{proof}
Rewriting (BP-2) we have
\[q = \frac{p-\pi_0}{\pi_1-\pi_0}\]
where $\ltau \leq \pi_0\leq p \leq \pi_1 \leq \utau$. Define
\[\underline{q}(p) = \inf_{\pi_0,\pi_1 \in [\ltau,\utau]} \frac{p-\pi_0}{\pi_1-\pi_0} \quad \mbox{s.t. } \pi_0\leq p \leq \pi_1 \mbox{ and } \pi_0 < \frac12 \leq \pi_1 \]
\[\overline{q}(p) = \sup_{\pi_0,\pi_1 \in [\ltau,\utau]} \frac{p-\pi_0}{\pi_1-\pi_0} \quad \mbox{s.t. } \pi_0\leq p \leq \pi_1  \mbox{ and } \pi_0 < \frac12 \leq \pi_1.\]
Then
\[\underline{q}(p) = \left\{ \begin{array}{cc}
0 & \mbox{ if } p \in [0,1/2] \\
\frac{2p-1}{2\utau - 1} & \mbox{ if } p \in (1/2,1]
\end{array}\right.\]
This value is attained exactly by $(\pi_1,\pi_0)=(\utau,p)$ on $p \in [0,1/2)$, and is reached in the limit $\pi_1 \uparrow \frac12$,  $\pi_0 \uparrow \frac12$ at $p=1/2$, and in  the limit   $\pi_0 \uparrow 1/2$,  $\pi_1=\utau$ on $p \in (1/2,1]$. 

Similarly, 
\[\overline{q}(p) = \left\{ \begin{array}{cc}
\frac{2(p-\ltau)}{1-2\ltau} & \mbox{ if } p \in [0,1/2] \\
1 & \mbox{ if } p \in (1/2,1]
\end{array}\right.\]
which is attained by $(\pi_1,\pi_0)=(1/2,\ltau)$ on $p \in [0,1/2]$, and $(\pi_1,\pi_0)=(p,\ltau)$ on $p \in (1/2,1]$.
\end{proof}

\subsection{Solving the problem for a fixed $q$} The generalized versions of $\phi_W$ and $\phi_B$ from Appendix \ref{sec:phi} are
\begin{align*}
   \phi_W^\tau(q) = \min_{\pi_1^*,\pi_0^* \in [\ltau,\utau]} q(2\pi_1^*  - 1) \quad
   \mbox{s.t. } q\pi_1^* + (1-q) \pi_0^* = p
\end{align*}
and
\begin{align*}
   \phi_B^\tau(q) = \max_{\pi_1^*,\pi_0^* \in [\ltau,\utau]} q(2\pi_1^* - 1) 
   \quad \mbox{s.t. } q\pi_1^* + (1-q) \pi_0^* = p
\end{align*}
Their solutions are
\[\phi^\tau_B(q) = \left\{\begin{array}{cc}
q(2\utau - 1) & \mbox{ if } q\leq \hat{q}_\tau(p)  \\ 2p-2(1-q)\ltau -q & \mbox{ if } q> \hat{q}_\tau(p) \end{array}\right.\]
where $\hat{q}_\tau(p) = \frac{p-\ltau}{\utau-\ltau}$, and
\[\phi^\tau_W(q) = \left\{\begin{array}{cc}
q(2\ltau - 1) & \mbox{ if } q\leq 1-\hat{q}_\tau(p)  \\
2p-2(1-q)\utau - q & \mbox{ if } q> 1-\hat{q}_\tau(p). \end{array}\right.\]

As $q$ varies from 0 to 1, the set $(\phi^\tau_W(q),\phi^\tau_B(q))$ takes on one of the four forms depicted in Figure \ref{fig:phitau}. By simple algebra,
\begin{equation} \label{eq:Comparep}
\frac{\ltau+\utau}{2} \geq p \quad \Longleftrightarrow \quad \hat{q}_\tau(p) \leq 1-\hat{q}_\tau(p).
\end{equation}
The best-case payoff is always initially increasing in $q$ and eventually decreasing, while the worst-case payoff is always  initially decreasing in $q$ and eventually increasing. But when $\hat{q}_\tau(p) < 1- \hat{q}_\tau(p)$ (Panels (a) and (b)), the best-case payoff reaches its peak and starts decreasing before the worst-case payoff hits its minimum and starts improving. If instead $\hat{q}_\tau(p) > 1- \hat{q}_\tau(p)$ (Panels (c) and (d)), the worst-case payoff begins improving before the best-case payoff starts falling. \footnote{As before, the relationship between $p$ and $1/2$ determines whether the no-information payoff is zero (treating no patients) or $2p-1$ (treating all patients). In Panels (a) and (d), the distrust vector is the origin (achieved at $q=0$), while in Panels (b) and (c) the distrust vector is $(2p-1,2p-1)$ (achieved $q=1$).} 
\begin{figure}[h]
\begin{center}
    \includegraphics[scale=0.35]{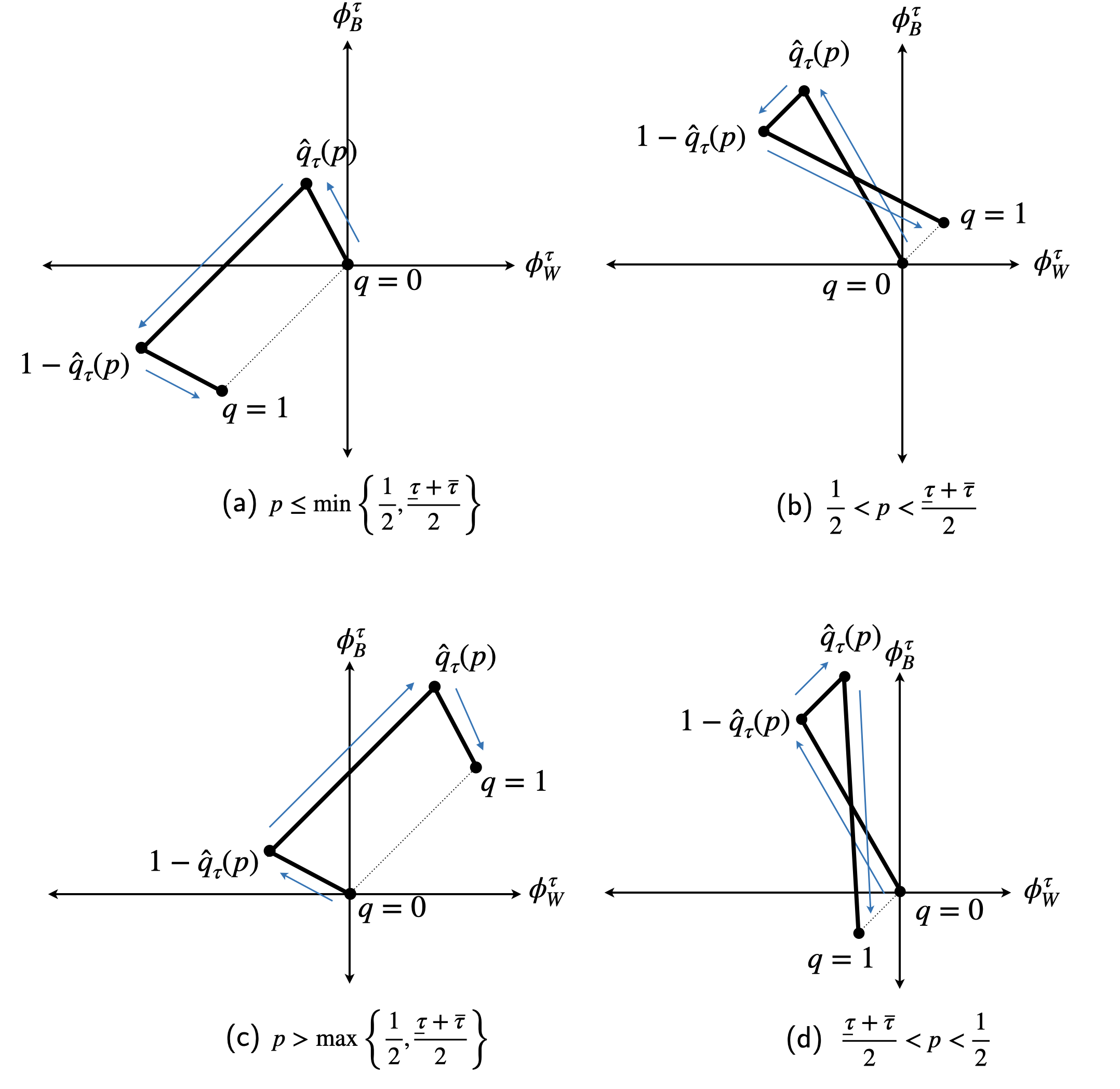}
    \end{center}
    \caption{Depiction of $(\phi^\tau_W(q),\phi^\tau_B(q))$ as $q$ varies from 0 to 1, for different ranges of $p$.} \label{fig:phitau}
\end{figure}

In our previous analysis the two thresholds of $1/2$ and $(\ltau+\utau)/2$ coincided, so we had only the analogs of cases (a) and (c). In this more general setting, the cases (b) and (d) emerge.

\subsection{Sufficiency of two-point supports.} The argument proceeds identically to Lemma \ref{lemm:p} and so is omitted.

\subsection{Undominated points.} Since we've restricted to $p<1/2$, we can focus on cases (a) and (d). 

Suppose we are in case (a), i.e.,  $p \leq \min \left\{\frac12, \frac{\ltau+\utau}{2}\right\}$.  Recall by Lemma \ref{lemm:BPtau} that for $p<1/2$ the set of incentive-compatible values of $q$ is $[0, 2 \frac{p-\underline{\tau}}{1-2\underline{\tau}}]$. As before, we first determine the largest value of $q$ at which $1-q$ is also incentive compatible, i.e, 
\[\overline{q}_\tau(p) = \sup\{ q\in [0,1]: 1- q \in BP_\tau(p)\}.\]
This generalizes the previous threshold $\overline{q}(p)$ from our earlier proof, and is solved by
\[\overline{q}_\tau(p) := \frac{1-2p}{1-2\underline{\tau}}.\]
Then if $\mu_{q'}=(q',1-q')$ is the two-support distribution corresponding to $X_1 \sim \mbox{Bernoulli}(q)$,  the set of feasible treatment probabilities $q$ is
\begin{equation} \label{eq:Q}
\mathcal{Q}(p,\mu_{q'})=\left\{\begin{array}{cc}
\{0,q'\} & \mbox{ if } q'  < \overline{q}_\tau(p) \\
\{0,q',1-q'\} & \mbox{ if } q' \in [\overline{q}_\tau(p),1/2]
\end{array}\right.
\end{equation}
When $\overline{q}_\tau(p) > \hat{q}_\tau(p)$, then as argued in the proof of Lemma \ref{lemm:p}, every point on the segment 
 \[\left\{(q'(2\ltau - 1),q'(2\utau-1)): q' \in [0, \hat{q}(p\right)]\}\]
 is implemented by $X_1 \sim \mbox{Bernoulli}(q')$ and this constitutes the frontier. This is because the feasible values of $q$ are $\{0,q'\}$, and both the aligned and misaligned AI optimally choose to recommend treatment with probability $q'$, yielding the worst-case payoff $q'(2\utau-1)$ and the best-case payoff $q'(2\utau-1)$. Since
 \[\overline{q}_\tau(p) > \hat{q}_\tau(p) \quad \Longleftrightarrow \quad \frac{p-\ltau}{\utau-\ltau} > \frac{1-2p}{1-2\ltau} \quad \Longleftrightarrow \quad p < p^* := \frac{\utau - 2\ltau^2}{1+2\overline{\tau}-4\ltau}\]
 the conditions $p \leq \min \left\{\frac12, \frac{\ltau+\utau}{2}\right\}$ and $\overline{q}_\tau(p) > \hat{q}_\tau(p)$ can be equivalently stated as  $p \leq \min \left\{\frac12, \frac{\ltau+\utau}{2},p^*\right\}$. But  $p^* < \frac12$ can be restated as $(2\ltau-1)^2 \geq 0 $, which always holds. And $p^* \leq\frac{\ltau+\utau}{2}$ if and only if $ (\utau-\ltau)(2\utau-1)$, which holds under our assumptions  $\utau\geq 1/2$ and $\utau\geq \ltau$. Thus
 \begin{equation} \label{eq:pstar}
 p^*\leq 1/2 \quad \mbox{and} \quad p^*\leq\frac{\ltau+\utau}{2}
 \end{equation}
 So $p<p^*$ is both necessary and sufficient for $p \leq \min \left\{\frac12, \frac{\ltau+\utau}{2},p^*\right\}$. This yields Part (a) of the result.

Now suppose we are in case (a) but $\hat{q}_\tau(p) > \overline{q}_\tau(p)$. Equivalently, suppose that  $p^* \leq p \leq \min \left\{\frac12,\frac{\ltau+\utau}{2}\right\}$. In this case the frontier consists of two line segments. The first segment is
  \[\left\{(q'(2\ltau - 1),q'(2\utau-1)): q' \in [0, \overline{q}_\tau(p\right)\}\]
   where each point is implemented by $X_1 \sim \mbox{Bernoulli}(q')$ for the corresponding $q'$. Here both the aligned and misaligned AI induce treatment for a $q'$-fraction of the population. This line segment connects the origin to the hedge point $H_\tau = (\overline{q}_\tau(p)(2\ltau-1),\overline{q}_\tau(p)(2\utau-1)).$ The second line segment connects $H_\tau$ to the trust point
 \begin{equation}
     \label{eq:TrustSplit}
T_\tau=\left( (1-\hat{q}_\tau(p))(2\ltau -1), \hat{q}_\tau(p)(2\utau-1) \right) \end{equation}
which is implemented by $X_1 \sim \mbox{Bernoulli}(\hat{q}_\tau(p))$. Here the aligned AI induces treatment for a $\hat{q}_\tau(p)$-fraction of patients, while the misaligned AI induces treatment for a $(1-\hat{q}_\tau(p))$-fraction of patients. The points on the line segment  $\overline{H_\tau T_\tau}$ are implemented by randomizing between the endpoints. That the points on $\overline{DH_\tau}\cup\overline{H_\tau T_\tau}$ points are undominated, and are the only undominated points, follows by the same reasoning as in the proof of Lemma \ref{lemm:p}.

Finally consider case (d), i.e., $\frac{\ltau+\utau}{2}<p<\frac12$. By (\ref{eq:Comparep}),
\begin{equation} \label{eq:hatq}
\frac{\ltau+\utau}{2}<p \quad \Longrightarrow \quad 1-\hat{q}_\tau(p) < \hat{q}_\tau(p).
\end{equation}
Moreover,
\[\hat{q}_\tau(p) = \frac{p-\ltau}{\utau-\ltau} \leq \frac{2(p-\ltau)}{1-2\ltau} \quad \Longleftrightarrow \quad \utau \geq 1/2\]
where the RHS is satisfied by assumption. Thus for $X_1 \sim \mbox{Bernoulli}(\hat{q}_\tau(p))$, the set of feasible choices of $q$ includes both $\overline{q}_\tau(p)$ and also $1-\overline{q}_\tau(p)$. The aligned AI picks $\hat{q}_\tau(p)$ and the misaligned AI picks $1-\hat{q}_\tau(p)$, yielding the trust point $T_\tau$ as defined in (\ref{eq:TrustSplit}). All points on the line between $T_\tau$ and $D=(0,0)$ can be implemented by randomization.

It remains to verify these points constitute the frontier. The slope of the $\overline{T_\tau D}$ line segment is
\[\frac{\hat{q}_\tau(p)(2\utau-1)}{(1-\hat{q}_\tau(p))(2\ltau-1)}  = \frac{p-\ltau}{\utau-p} \cdot \frac{2\utau-1}{2\ltau-1} \geq -1 \]
where the final inequality follows from our assumption that $p<1/2$.

For any $1-\hat{q}_\tau(p) < q< 1/2$, if  $X_1 \sim \mbox{Bernoulli}(q)$ then both $q$ and $1-q$ are feasible by Lemma \ref{lemm:BPtau}. The aligned AI induces treatment of fraction $1-q$ of the population, while the misaligned AI induces treatment of fraction $q$ of the population, yielding the point on the line connecting $T_\tau$ to 
\begin{align*}
B := \left(\phi^\tau_W \left(\frac12\right),\phi^\tau_B\left(\frac12\right)\right) = \left(2p-\utau-\frac12,\utau-\frac12\right).\end{align*}
The slope of the line segment $\overline{T_\tau B}$ is $-1$,
so $\overline{T_\tau B}$ lies everywhere below $\overline{T_\tau D}$ and is thus dominated.

If instead $\overline{q}_\tau(p) \leq q \leq 1-\hat{q}_\tau(p)$ then the aligned AI still induces treatment of fraction $1-q$ and the misaligned AI still induces treatment of fraction $q$, but this now yields a point on the line connecting $T_\tau$ to 
\begin{align*}
    C & := (\phi^\tau_W(\overline{q}_\tau(p)),\phi^\tau_B(1-\overline{q}_\tau(p))) \\
    & = (\overline{q}_\tau(p)(2\ltau-1), 2p-2\overline{q}_\tau(p)\ltau - (1-\overline{q}_\tau(p))) \\
    & = \left(2p-1,0\right)
    \end{align*}
The slope of $\overline{T_\tau C}$ is again $-1$, so it too lies everywhere below $\overline{T_\tau D}$.

And finally, if $ q < \overline{q}_\tau(p)$, then both the aligned and misaligned AI induce treatment of a fraction $q$ of the population (since $1-q$ is not incentive-compatible), thus yielding the point
\[(q(2\ltau-1),q(2\utau-1))\]
But the slope of the line connecting this point to the origin is
\[\frac{2\utau-1}{2\ltau-1}  > \frac{\hat{q}_\tau(p)(2\utau-1)}{(1-\hat{q}_\tau(p))(2\ltau-1)} \]
where the inequality follows from (\ref{eq:hatq}). So these points again lie below the purported frontier. This concludes the proof.

\subsection{The $p>1/2$ case} \label{sec:p>0.5}

The proof of the characterization for $p>\frac12$ follows identically and yields the following characterization. Define
\[p^{(3)}_\tau=\frac{2\utau + 2(\utau-1)\ltau(\utau-\ltau)}{2\utau-1)(\utau-\ltau)+2}\]
The trust point is  
\[T^\tau_p = \left\{\begin{array}{cc}
(w_\tau(p,\hat{q}), b_\tau(p,\hat{q})) & \mbox{ if }p > p^*_2 \\
(w_\tau(p,1-\hat{q}), b_\tau(p,\hat{q}))& \mbox{ if } p^*_2 > p \\
\end{array}\right.
\]
The hedge point is
$H_\tau=
\left(w_\tau\!\left(q_H\right),b_\tau\!\left(q_H\right)\right)$ where 
$q_H = 
\frac{2\utau-2p}{2\utau-1}$.

\begin{lemma} Suppose $p > 1/2$ and $\ltau >\frac12$. Then the frontier is simply the distrust point $D_p$.
\end{lemma}

In this trivial case, the AI cannot induce the designer to assign probability less than $1/2$ to need of treatment, so the designer always treats.

\begin{lemma} Suppose $p > 1/2$ and $\ltau \leq 1/2$. Then:
\begin{itemize}
    \item[(a)] If $p\notin \left[p^{(2)}_\tau,p^{(3)}_\tau\right]$, the efficient frontier is the line segment connecting the distrust point $D$ to the trust point $T_\tau$.
        \item[(b)] If $p\in \left(p^{(2)}_\tau,p^{(3)}_\tau\right)$, the efficient frontier is the union of the line segment connecting $D$ to $H_\tau$ and the line segment connecting $H_\tau$ to $T_\tau$.
\end{itemize}
\end{lemma}

\begin{lemma} Suppose $p > 1/2$ and $\ltau \leq 1/2$. Then:
\begin{itemize}
    \item[(a)] The trust point $T_p^\tau$ is implemented by $X_1 \sim \mbox{Ber}\left(\frac{p-\ltau}{\utau -\ltau}\right)$.
    \item[(b)] The hedge point $H_p^\tau$ is limit-implemented by $X_1 \sim \mbox{Ber}\left(\frac{2\utau-2p}{2\utau-1}\right)$.
    \item[(c)] The distrust point $D$ is implemented by any constant $X_1$.
\end{itemize}
\end{lemma}

\section{Proof of Proposition \ref{prop:OptimalTau}}

Fix $p \in [0,1/3]$ and a preference parameter $\alpha \in [0,1]$.  Let
\[\mathcal{T} = \{(\ltau,\utau) \in [0,1]^2 : 0 \leq \ltau \leq p \leq \utau \leq 1\}\]
be the set of feasible values of $(\ltau,\utau)$. This set is partitioned by the thresholds 
\[p^*_1(\ltau,\utau) := \frac{\utau - 2\ltau^2}{1+2\utau - 4\ltau} \leq \frac{\ltau+\utau}{2} =: p^{(2)}_\tau(\ltau,\utau) \]
which define three regions:
\begin{align*}
\mathcal{T}_1  & = \left\{(\ltau,\utau) : p\leq p^*_1(\ltau,\utau)\right\} \cap \mathcal{T} \\
\mathcal{T}_2 & = \left\{(\ltau,\utau) : p^*_1(\ltau,\utau) < p \leq  p^*_2(\ltau,\utau)\right\} \cap \mathcal{T} \\ 
\mathcal{T}_3 & =  \left\{(\ltau,\utau) : p^*_2(\ltau,\utau) < p \right\}\cap \mathcal{T}.
\end{align*}

On $\mathcal{T}_1$, the frontier is the line segment connecting the distrust point $D = (0,0)$ to the Case 1 trust point
\[T^{(1)}_p(\tau) = \left(\frac{p-\ltau}{\utau-\ltau}\right)(2\ltau-1,2\utau-1)\]
On $\mathcal{T}_2$, the frontier is the union of the line segment connecting the distrust point $D$ to the hedge point
\[H_p(\tau) = \left(\frac{1-2p}{1-2\ltau}\right)(2\ltau-1,2\utau-1),\]
and the line segment connecting $H_p(\tau)$ to the Case 2 trust point
\[T^{(2)}_p(\tau) = \left(\left(1 - \frac{p-\ltau}{\utau-\ltau}\right)(2\ltau-1), \left(\frac{p-\ltau}{\utau-\ltau}\right)(2\utau-1)\right).\]
Finally, on $\mathcal{T}_3$, the frontier is the line segment connecting the distrust point $D$ to the trust point $T^{(2)}_p(\tau)$. 

The designer's payoff function is
\[U_\alpha(\underline{v},\overline{v}) = \alpha \underline{v} + (1-\alpha)\overline{v}.\]
Since this preference is linear, it is maximized at one of the extreme points of the frontier. To save on notation, define 
$U_T^1(\ltau,\utau) = U_\alpha\left(T_p^{(1)}(\tau)\right)$, $U_T^2(\ltau,\utau) = U_\alpha\left(T_p^{(2)}(\tau)\right)$, and $U_H(\ltau,\utau) = U_\alpha(H_p(\tau))$, noting that the payoff from the distrust point is always zero. Then the designer's problem can be rewritten as
\begin{equation} \label{eq:Max}
\max \left\{ \max_{\tau \in \mathcal{T}_p^1} U_T^1(\tau), \max_{\tau \in \mathcal{T}_p^2} U_T^2(\tau), \max_{\tau \in \mathcal{T}_p^3} U_T^2(\tau), \max_{\tau \in \mathcal{T}_p^2} U_H(\tau), 0\right\}
\end{equation}

The plan of the proof is as follows. Section \ref{sec:Facts} collects preliminary results about how the payoffs at the extreme points vary in $\ltau$ and $\utau$. Section \ref{sec:T1} characterizes the optimal solution supposing only the distrust point and case 1 trust point were available. Section \ref{sec:Eliminate} then shows that no other extreme points can be optimal, so that our result in Section \ref{sec:T1} is the global solution.

\subsection{Useful Facts} \label{sec:Facts}

\begin{fact} \label{fact:T2}
\begin{itemize}
\item[(a)] For any $\alpha \in [0,1/2)$ and $\ltau \in [0,1]$, the function $U_T^2(\ltau,\utau)$
is increasing in $\utau$; moreover, the function $U_T^2(\ltau,1)$ is decreasing in $\ltau$.
\item[(b)] For any $\alpha \in [1/2,1]$ and $\ltau \in [0,1]$, the function $U_T^2(\ltau,1)$
is decreasing in $\utau$.
\end{itemize}
\end{fact}

\begin{proof} By algebra,
\[\frac{\partial U_T^2(\ltau,\utau)}{\partial \utau} = \frac{(2\alpha - 1) (2\ltau-1)(p-\ltau)}{(\utau-\ltau)^2}\]
Because $p-\ltau>0$ and $2\ltau-1<0$, $\sgn\left(\frac{\partial U_T^2(\ltau,\utau)}{\partial \utau}\right) = \sgn(1-2\alpha)$. So for $\alpha<1/2$, $U_T^2(\ltau,\utau)$ is increasing in $\utau$, and for $\alpha\geq 1/2$, it is decreasing in $\utau$.

Also
\[\frac{\partial U_T^2(\ltau,1)}{\partial \ltau} = \frac{(2\alpha - 1)(1-p)}{(\utau-\ltau)^2}\]
which is positive for $\alpha\geq 1/2$ and negative for $\alpha<1/2$.
\end{proof}

\begin{fact} \label{fact:H} For every $\ltau \in [0,1]$, the function $U_H(\ltau,\utau)$ is increasing in $\utau$; moreover, the function $U_H(\ltau,1)$ is increasing in $\ltau$.
\end{fact}

\begin{proof} By algebra,
\[\frac{\partial U_H(\ltau,\utau)}{\partial \utau} = 2(1-\alpha) \frac{1-2p}{1-2\ltau} >0\]
and
\[\frac{\partial U_H(\ltau,1)}{\partial \ltau} =  \frac{2(1-\alpha)(1-2p)}{2(\ltau-1)^2}>0\]
as desired.
\end{proof}

\begin{fact} \label{fact:T1} For every $\ltau \in [0,1]$, the function 
$U_T^1(\ltau,\utau)$
is increasing in $\utau$.
\end{fact}

\begin{proof} By algebra
\[\frac{\partial U_T^1(\ltau,\utau)}{\partial \utau} = \frac{(p-\ltau)(1-2\ltau)}{(\utau-\ltau)^2}>0\]
since $p>\ltau$ and $1-2\ltau>0$.
\end{proof}

\subsection{The Case 1 Optimum} \label{sec:T1}

\begin{lemma} \label{lemm:T1} Suppose $p<1/3$ and define
$\ltau^\circ= 1-\sqrt{\frac{1-p}{2\alpha}}$ 
and
\[\underline{\alpha} (p) := \frac{1-p}{2} < \frac{1}{2(1-p)} =: \overline{\alpha}(p)\]
Then 
\[\argmax_{(\ltau,\utau) \in \mathcal{T}_1}U_T^1(\ltau,\utau) = \left\{ \begin{array}{cc}
(0,1) & \mbox{ if } \alpha < \underline{\alpha}(p)\\
(\underline{\tau}^\circ,1) & \mbox{ if }  \underline{\alpha}(p) < \alpha < \overline{\alpha}(p) \\
(p,1) & \mbox{ if } \alpha > \overline{\alpha}(p)
\end{array}\right.\]
yielding the payoffs
\[\max_{(\ltau,\utau) \in \mathcal{T}_1}U_T^1(\ltau,\utau) = \left\{ \begin{array}{cc}
p(1-2\alpha) & \alpha < \underline{\alpha}(p)\\
\frac{p-\underline{\tau}^\circ}{1-\underline{\tau}^\circ}(1+2\alpha (\underline{\tau}^\circ-1)) & \mbox{ if } \underline{\alpha}(p) < \alpha < \overline{\alpha}(p) \\
0 & \mbox{ if } \alpha > \overline{\alpha}(p)
\end{array}\right.\]
\end{lemma}

\begin{proof}
The threshold $p^*_1(\ltau,\utau)$ is increasing in $\utau$ for every $\ltau$. Thus if $(\ltau, \utau) \in \mathcal{T}_1$ then also $(\ltau,1) \in \mathcal{T}_1$. Moreover, $U_T^1(\ltau,1) \geq U_T^1(\ltau,\utau)$ by Fact \ref{fact:T1}. So the designer sets $\utau^*=1$ and optimizes over $\ltau$. The payoff criterion simplifies to
\[\left(\frac{p-\ltau}{1-\ltau}\right)(2\alpha(\ltau-1)+1) := g(\ltau)\]
whose derivative is
\[g'(\ltau) = \frac{2 \alpha \ltau^2 -4\alpha \ltau +2\alpha +p-1}{(1-\ltau)^2}\]
This derivative is positive up until the critical point
\[\ltau^\circ= 1-\sqrt{\frac{1-p}{2\alpha}}\]
and negative after. So it remains to determine if $\underline{\tau}^\circ$ falls in the feasible region of values for $\ltau$.

Plugging in $\utau^*=1$, the inequality $p < p^*(\ltau,1)$ becomes
\[2\ltau^2 -4p\ltau +3p -1<0.\]
The two roots of this quadratic are 
\[\ltau_{\pm}(p) = \frac{4p \pm \sqrt{16 p^2 - 24p+8}}{4}\]
where $\ltau_-(p)< 0 < p < \ltau_+(p)$ for $p<1/3$. The quadratic is negative between the roots, so all $\ltau \in [0,p]$ are feasible.  Since finally $\ltau^\circ\geq 0$ if and only if $ \alpha \geq \frac{1-p}{2} := \alpha_0(p)$ and $\ltau^\circ \leq p$ if and only if $ \alpha \leq \frac{1}{2(1-p)} := \alpha_1(p)$, the result obtains.
\end{proof}

\subsection{The Other Extreme Points Cannot be Optimal} \label{sec:Eliminate}

 \begin{lemma}[The Case 2 Trust Point is Never Optimal]  \label{lemm:T2NotOptimal} For every $(\ltau,\utau) \in \mathcal{T}_2$, either $U_T^2(\ltau,\utau)<0$ or $U_T^2(\ltau,\utau) < \max_{(\ltau,\utau) \in \mathcal{T}_1} U_T^1(\ltau,\utau)$.
\end{lemma}

\begin{proof} 
Suppose $\alpha > 1/2$. Then
\begin{align*}
U_T^2(\ltau,\utau) & \leq U_T^2(\ltau,p) && \mbox{ by Fact \ref{fact:T2} Part (b)}\\
&= (1-\alpha)(2p-1) <0 && \mbox{ since $p<1/2$}
\end{align*}
Since the distrust point yields a payoff of zero, this cannot be optimal.

Suppose $\alpha \leq 1/2$. Then Fact \ref{fact:T2} Part (a) implies that
\[U_T^2(\ltau,\utau) \leq U_T^2(\ltau,1) \leq U_T^2(0,1)=p-\alpha\]
But for every $(p,\alpha)$ for which $\alpha < p $ (so that the payoff is positive), it also holds that $\alpha < \frac{1-p}{2}$ (since $p<1/3$ implies $p < \frac{1-p}{2}$). Thus Lemma \ref{lemm:T1} implies $\max_{\tau \in \mathcal{T}_1} U_T^1(\tau)=p(1-2\alpha)$. Moreover
\[p-\alpha < p(1-2\alpha)\]
for all $p < 1/3$. So the case 2 trust point cannot be optimal for any $\alpha$.

\end{proof}

\begin{lemma}[The Hedge Point is Never Optimal] \label{lemm:HedgeNotOptimal} For every $\tau \in \mathcal{T}_2$,  either $U_H(\tau) < 0$ or $U_H(\tau)< \max_{\tau \in \mathcal{T}_1} U_T^1(\tau)$.
\end{lemma}

\begin{proof}
Suppose $\alpha \geq \frac{2\utau-1}{2(\utau-\ltau)}$, implying $\alpha(2\ltau-1)+(1-\alpha)(2\tau-1) < 0$. Then
\[U_H(\ltau,\utau) = \left(\frac{1-2p}{1-2\ltau}\right)(\alpha(2\ltau-1)+(1-\alpha)(2\utau-1))<0 \]
since $\frac{1-2p}{1-2\ltau} \geq 0$ (by assumption that $p < 1/2$ and $\ltau < p$). So the hedge point is dominated by the distrust point.
 
 Alternatively, suppose $\alpha < \frac{2\utau-1}{2(\utau-\ltau)}$, in which case $\alpha(2\ltau-1)+(1-\alpha)(2\tau-1) \geq 0$. Then
\begin{align} \label{ineq:TbeatsH}
U_T^1(\ltau,\utau) &= \left(\frac{p-\ltau}{\utau-\ltau}\right)(\alpha(2\ltau-1)+(1-\alpha)(2\tau-1)) \\
& \geq \left(\frac{1-2p}{1-2\ltau}\right) (\alpha(2\ltau-1)+(1-\alpha)(2\tau-1)) = U_H(\ltau,\utau) \nonumber
\end{align}
since  $\frac{1-2p}{1-2\ltau} < \frac{p-\ltau}{\utau-\ltau}$  for all $(\ltau,\utau) \in \mathcal{T}_2$ (this inequality is equivalent to $p >p^*_1(\ltau,\utau)$). That is, for the \emph{same} $(\ltau,\utau)$ pair, the trust point characterization in $U_T^1$ yields a higher payoff than the hedge point characterization in $U_H$.

Finally, for any $(\ltau,\utau) \in \mathcal{T}_2$,
\begin{align*}
U_H(\ltau,\utau) & \leq U_T^1(\ltau,\utau) && \mbox{ by (\ref{ineq:TbeatsH})} \\
& \leq 
U_T^1(\ltau,1) && \mbox{ by Fact \ref{fact:T1}} \\
& \leq U_T^1(\ltau^\circ,1) && \mbox{ by optimality of $\ltau^\circ$ }
\end{align*}
where $(\ltau^\circ,1) \in \mathcal{T}_1$ as we showed in the proof of Lemma \ref{lemm:T1}.
\end{proof}
  
 \begin{lemma}[The Case 3 Trust Point is Never Optimal]  \label{lemm:T3NotOptimal} For every $(\ltau,\utau) \in \mathcal{T}_3$, either $U_T^2(\ltau,\utau)<0$, $U_T^2(\ltau,\utau) < \max_{(\ltau,\utau) \in \mathcal{T}_1} U_T^1(\ltau,\utau)$, or $U_T^2(\ltau,\utau) < \max_{(\ltau,\utau) \in \mathcal{T}_2} U_T^2(\ltau,\utau)$.
\end{lemma}

\begin{proof} 

As we've argued in the proof of Lemma \ref{lemm:T2NotOptimal}, when $\alpha>1/2$ then $U^T_2(\ltau,\utau)<0$ for all $(\ltau,\utau)$, so the case 3 trust point is dominated by the distrust point.

Suppose $\alpha <1/2$. For every $(\ltau,\utau) \in \mathcal{T}_3$, 
\[\utau \leq 2p-\ltau <1\]
since $p<1/2$. So $(\ltau,1)$ belongs to either $\mathcal{T}_1$ or $\mathcal{T}_2$, and 
$U_T^1(\ltau,1) \geq U_T^2(\ltau,1) \geq U_T^2(\ltau,\utau).$ Thus $U_T^2(\ltau,\utau)$ is dominated by one of the other trust points.
\end{proof}

\section{Alternative Designer Preferences} \label{app:AlternativePref}

Following \citet{hurw1951} we consider designer preferences that maximize a weighted sum of payoffs  when Nature is benevolent and when it is adversarial, with parameter $\beta \in (0,1)$ corresponding to the probability assigned to an adversarial Nature. To keep notation simple, we henceforth refer to a benevolent Nature as \emph{aligned} and an adversarial Nature as \emph{misaligned}.

Specifically, Table \ref{tab:Cases} describes four payoffs for a fixed covariate $X_1\sim \mu$ depending on whether each of Nature and the AI are aligned or misaligned. We maintain the requirement that the true distribution $P^*$ and the AI's selected $P$ must be admissible.

\begin{table}[h]
\begin{center}
\footnotesize
\begin{tabular}{|c|c|c|c|}
\hline
      &                & \multicolumn{2}{c|}{\textbf{Nature}} \\ \cline{3-4}
      &                & misaligned & aligned \\ \hline
\multirow{2}{*}{\textbf{AI}}
      & misaligned & \cellcolor{gray!20}{$\displaystyle v_{MM}(\mu) := \inf_{P^*\in\mathcal{P}(I,\mu)}
                 \inf_{P\in\mathcal{P}(I,\mu)}U(P,P^*)$}
             & $\displaystyle v_{AM}(\mu) := \sup_{P^*\in\mathcal{P}(I,\mu)}
                 \inf_{P\in\mathcal{P}(I,\mu)}U(P,P^*)$\\
                 \hline
      & aligned & $\displaystyle v_{MA}(\mu) := \inf_{P^*\in\mathcal{P}(I,\mu)}
                 \sup_{P\in\mathcal{P}(I,\mu)}U(P,P^*)$
             & \cellcolor{gray!20}{$\displaystyle v_{AA}(\mu):= \sup_{P^*\in\mathcal{P}(I,\mu)}
                 \sup_{P\in\mathcal{P}(I,\mu)}U(P,P^*)$} \\ \hline
\end{tabular}
\end{center}
\caption{} \label{tab:Cases}
\end{table}

\noindent The order of subscripts in $v_{ij}$ is such that $i$ indicates the alignment of Nature and $j$ indicates the alignment of the AI.

The new ``worst case'' payoff corresponds to certainty that the AI is misaligned

 \[\underline{v}_\beta(\mu) = \beta v_{MM}(\mu) +(1-\beta) v_{AM}(\mu)\]
 while the new ``best case'' corresponds to certainty that the AI is aligned
 \[\overline{v}_\beta(\mu) = \beta v_{MA}(\mu) + (1-\beta) v_{AA}(\mu).\]

 To keep things simple we consider a singleton $\mathcal{X}_0$ with $p := p_{x_0} <1/3$. 
 
 \begin{proposition}\label{prop:small p} When $p := p_{x_0} <1/3$, the efficient frontier is the line connecting the distrust point $D=(0,0)$ to the trust point $T_\beta = (-\beta p,(1-\beta)p).$
\end{proposition}

Thus the frontier is again a line segment connecting the distrust point to the trust point, although the slope of the segment now depends on the designer's belief parameter $\beta$. See Figure \ref{fig:beta}.

\begin{figure}[h]
    \centering
    \includegraphics[scale=0.3]{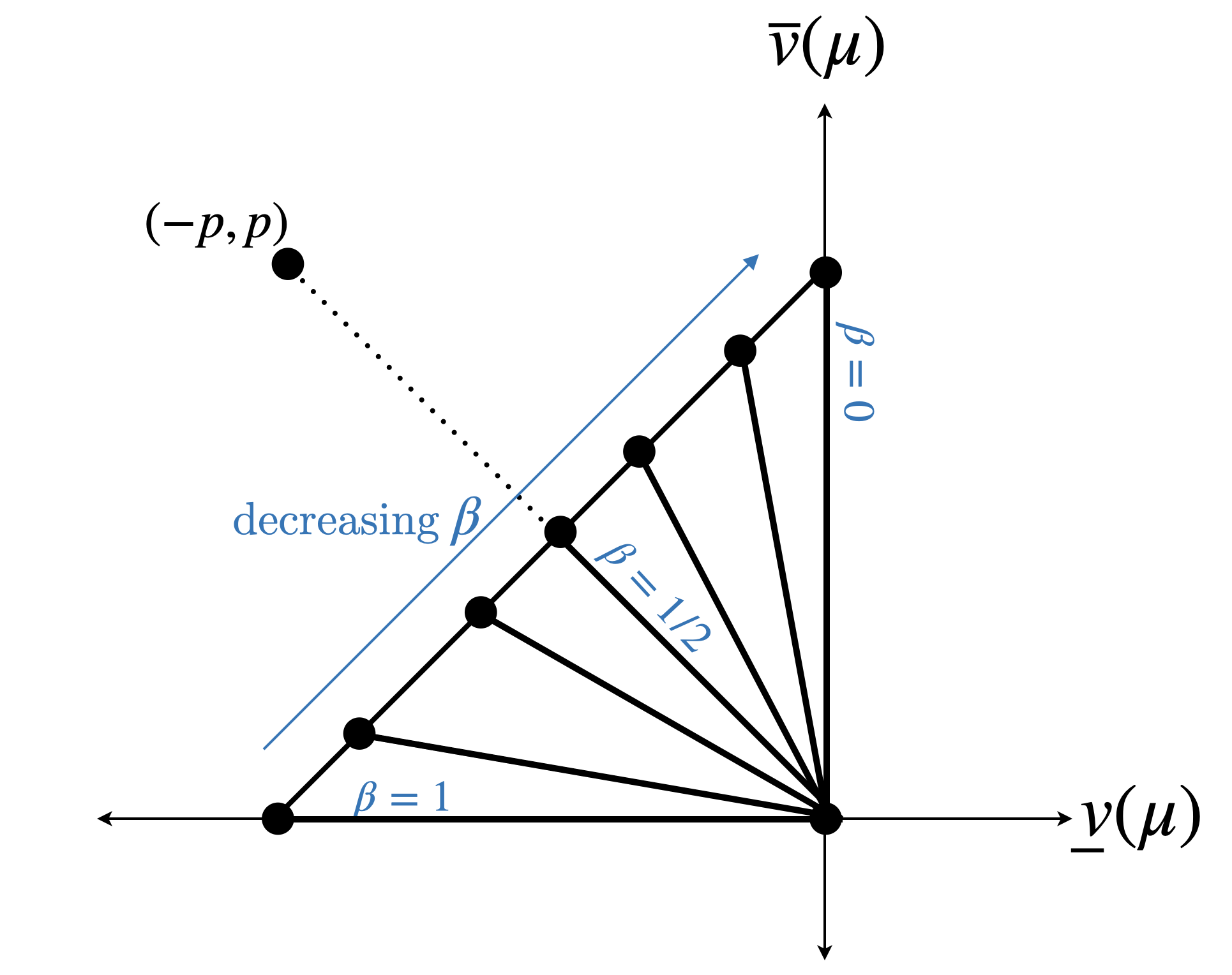}
    \caption{Comparison of the frontier for different values of $\beta$}
    \label{fig:beta}
\end{figure}

 We will write $\mathcal{P}(p,\mu)$ for the set of admissible distributions given prior probability of need of treatment $p=P(Y=1)$ and covariate $X_1 \sim \mu$. We have already shown that when $p <1/3$, 
 \[v_{MM}(\mu) = \inf_{P^*\in\mathcal{P}(p,\mu)} \inf_{P\in\mathcal{P}(p,\mu)} U(P,P^*) = -p\]
 while
 \[v_{AA}(\mu)= \sup_{P^*\in\mathcal{P}(p,\mu)} \sup_{P\in\mathcal{P}(p,\mu)} U(P,P^*)=p.\]

 Section \ref{sec:OffDiagonal} proves results about the remaining two cells of Table \ref{tab:Cases}, and Section \ref{sec:TwoPoints} uses these to complete the proof.

 \subsection{Preliminary Results} \label{sec:OffDiagonal}

 \begin{lemma} \label{lemm:EG} For every $\mu$ and $p<1/3$, \[v_{MA}(\mu) = \inf_{P^*\in\mathcal{P}(p,\mu)} \sup_{P\in\mathcal{P}(p,\mu)} U(P,P^*) = 0.\]
 \end{lemma}

 That is, the designer's expected payoff from an adversarial Nature and aligned AI is always zero. This is because the aligned AI will leverage any information in $X_1$ to help the designer, and so the adversarial Nature chooses $X_1$ to be completely uninformative. 

\begin{proof} For any $P^*$ chosen by Nature, the aligned AI best responds by setting $P=P^*$. So the problem reduces to solving
\begin{align} 
\inf_{P^* \in \mathcal{P}(p,\mu)} U(P^*,P^*) &  = \mathbb{E}_{P^*}[u(Y,a(X_1,P^*))] \nonumber \\
& = \sum_{x_1 \in A} \mu(x_1) (2\mathbb{E}_{P^*}[Y=1 \mid X_1=x_1]-1) \label{eq:Naturepayoff}
\end{align}
where
\[A = \{x_1 \in \mathcal{X}_1 \mid \mathbb{E}_{P^*}(Y=1 \mid X_1=x_1) \geq 1/2\}\]
is the event on which the designer treats. 
Clearly (\ref{eq:Naturepayoff}) is at least zero, since $x_1\in A$ if and only if $2\mathbb{E}_{P^*}(Y=1 \mid X_1=x_1)-1\geq 0$. To see that zero is attainable, observe that Nature can pick $P^*$ such that $P^*(Y=1 \mid X_1=x_1)=p$ for all $x_1 \in \mathcal{X}_1$. This corresponds to a completely uninformative $X_1$ and leads the designer to treat no one.  

\end{proof}

  \begin{lemma} \label{lemm:EB} Fix any $p<1/3$ and $X_1 \sim \mbox{Ber}(q)$. Then \begin{equation} \label{eq:infsup} v_{AM}(\mu) = \sup_{P^*\in\mathcal{P}(p,\mu)} \inf_{P\in\mathcal{P}(p,\mu)} U(P,P^*) \leq 0
  \end{equation}
  where the value of zero is attained when $ q \leq p$.
 \end{lemma}

The lemma says that when $X_1 \sim \mbox{Bernoulli}(q)$, then the designer's expected payoff with an adversarial AI and benevolent Nature is weakly negative. The upper bound of zero is because an adversarial AI can always simply report that $X_1$ is uninformative, in which case the designer never treats. We show that when $q\in [0,p]$ then a payoff of zero can be attained: the designer chooses the conditional probability of need of treatment to be 1 on $X_1=1$, so that treating on this event yields a positive payoff. Since the AI cannot induce treatment on $X=0$ when $q <p$ (by our assumption that $p<1/3$), the adversarial AI's best response is to report that $X_1$ is uninformative, even though it is not. (Note that if Nature were to instead choose $X_1$ to be uninformative, the misaligned AI could induce a negative payoff for the principal by inducing treatment on $X_1=1$, so this is not optimal for a benevolent Nature.)

 \begin{proof}
As explained above, $v_{AM}(\mu) \leq 0$ because an adversarial AI can report that $X_1$ is uninformative. To see why $v_{AM}(\mu)=0$ when $q \leq p$ , observe that the optimization problem in (\ref{eq:infsup}) is equivalent to the following game: There is a given probability space $(\mathcal{X}_1,\mu)$ where $\mathcal{X}_1 = \{0,1\}$ and $\mu(X_1=1)=q<p$.  Nature first chooses a random variable $f: \mathcal{X}_1 \rightarrow [0,1]$ mapping each $x_1 \in \mathcal{X}_1$ to a conditional probability of need of treatment. The misaligned AI then picks an event $A \subseteq \mathcal{X}_1$ such that $\mu(A) \leq 2p$ (so that inducing treatment on $A$ is incentive compatible) in order to minimize $\sum_{x_1 \in A} \mu(x_1) (2f(x_1)-1)$. Nature's payoff is $\sum_{x_1 \in A} \mu(x_1) (2f(x_1)-1)$.

To attain a payoff of zero, Nature chooses
\[f(x_1) = \left\{\begin{array}{cc}
1 & \mbox{ if } x_1 =1 \\
\frac{p-q}{1-q} & \mbox{ if } x_1 =0
\end{array}\right.\]
where $\frac{p-q}{1-q} \in [0,1]$, so that the designer should optimally treat on $x_1=1$ and not on $x_1=0$. Since $\mu(X_1=1)=q<p$, then also $\mu(X_1=0)=1-q > 1-p$. But $1-p \geq 2p$ by assumption that $p<1/3$. So the AI can only choose $A$ to include $X_1=1$ and not $X_1=0$. Thus the AI chooses between giving Nature a payoff of zero (via $A = \varnothing$) or $q>0$ (via $A=\{1\}$), and minimizes this by choosing $A=\varnothing$.  

\end{proof}

\subsection{Completing the Proof of Proposition \ref{prop:small p}} \label{sec:TwoPoints}
 We will finally argue that it is without loss to restrict attention to two-point distributions $\mu_q=(q,1-q)$ with $q\in [0,p]$. This argument is separated into two parts: Lemma \ref{lemm:twopoint} shows that it is without loss to consider two-point distributions, and Lemma \ref{lemm:q} shows that it is without loss to suppose $q\in [0,p]$.
 
 \begin{lemma} \label{lemm:twopoint} For every distribution $\mu$ there is some $\mu_q=(q,1-q)$ such that $(\underline{v}_\beta(\mu_q),\overline{v}_\beta(\mu_q))$ dominates $(\underline{v}_\beta(\mu),\overline{v}_\beta(\mu))$.
 \end{lemma}

\begin{proof}

Consider any distribution $\mu$, and as in Section \ref{sec:2point} let $q^* \in \mathcal{Q}(p,\mu)$ be the feasible treatment probability that the aligned AI would pick when Nature is benevolent. Define $\mu_{q^*}=(q^*,1-q^*)$. Then $v_{AA}(\mu)=v_{AA}(\mu_{q^*})$ by construction, and we showed in Section \ref{sec:2point} that $v_{MM}(\mu) \leq v_{MM}(\mu_{q^*})$. Moreover, Lemma \ref{lemm:EG} implies $v_{MA}(\mu) = v_{MA}(\mu_{q^*})$, since its value of zero does not depend on the distribution.

It remains to argue that $v_{AM}(\mu) \leq v_{AM}(\mu_{q^*}).$ Since $q^*$ is an attainable event probability, there is some event $E$ such that $\mu(E)=q^*$. Nature's choice of $P^*$ is equivalent to choosing a function $f:\mathcal{X}_1 \rightarrow [0,1]$ such that each $f(x_1)$ is the conditional probability of need of treatment at state $x_1$.

Now let $IC(\mu)$ denote the set of subsets of $\mathcal{X}_1$ on which treatment can be induced in an incentive-compatible way, i.e., 
\[IC(\mu) = \{ A \subseteq \mathcal{X}_1 \mid \mu(A) \leq 2p.\}\]
For each such $A$, inducing treatment on $A$ leads to a payoff of 
\[w_f(A)= \sum_{x_1 \in A} \mu(x_1) (2f(x_1)-1).\] The misaligned AI chooses $A$ to minimize this payoff, so Nature's payoff for any given choice of $f$ is $\min_{A \in IC(\mu)} w_f(A)$. But if Nature were instead to pick $\mu_{q^*}$---coarsening the $\sigma$-algebra to $\{\varnothing, E,E^c,\mathcal{X}_1\}$---then treatment could only be induced on
\[IC(\mu_{q^*}) = \{ A \in \{\varnothing, E,E^c,\mathcal{X}_1\} \mid \mu(A) \leq 2p\},\]
Importantly, $IC(\mu_{q^*}) \subseteq IC(\mu)$, so for every $f$ and $\mu$
\[\min_{A \in IC(\mu)} w_f(A) \leq \min_{A \in IC(\mu_{q^*})} w_f(A)\]
implying
\[v_{AM}(\mu) = \max_{f: \mathcal{X}_1 \rightarrow [0,1]} \left[\min_{A \in IC(\mu)} w_f(A)\right] \leq \max_{f: \mathcal{X}_1 \rightarrow [0,1]} \left[\min_{A \in IC(\mu_{q^*})} w_f(A)\right] = v_{AM}(\mu_{q^*} \]
as desired. So
\begin{align*}
\underline{v}_\beta(\mu) & = \beta v_{MM}(\mu) +(1-\beta) v_{AM}(\mu) \\
& \leq \beta v_{MM}(\mu_{q^*}) +(1-\beta) v_{AM}(\mu_{q^*}) = \underline{v}_\beta(\mu_{q^*})
\end{align*}
while
\begin{align*}
\overline{v}_\beta(\mu) & = \beta v_{MA}(\mu) +(1-\beta) v_{AA}(\mu) \\
& = \beta v_{MA}(\mu_{q^*}) +(1-\beta) v_{AA}(\mu_{q^*}) = \overline{v}_\beta(\mu_{q^*})
\end{align*}
implying that $(\underline{v}_\beta(\mu),\overline{v}_\beta(\mu))$ is dominated by $(\underline{v}_\beta(\mu_{q^*}),\overline{v}_\beta(\mu_{q^*}))$. Thus as in the proof of Lemma \ref{lemm:p} it is sufficient to restrict attention to two-point distributions $\mu_q$. 
\end{proof}

\begin{lemma} \label{lemm:q} For every distribution $\mu_q$ there is some $q'\in [0,p)$ such that $(\underline{v}_\beta(\mu_{q'}),\overline{v}_\beta(\mu_{q'}))$ weakly dominates $(\underline{v}_\beta(\mu_q),\overline{v}_\beta(\mu_q))$.
\end{lemma}

\begin{proof} Pick an arbitrary $q$, where we can restrict $q \in [0,1/2]$  without loss by symmetry. If $q \leq p$ then the statement trivially follows by setting $q'=q$. Otherwise from our proof of Lemma \ref{lemm:p}, $v_{MM}(\mu_q) < v_{MM}(\mu_p)$ and $v_{AA}(\mu_q) < v_{AA}(\mu_p)$. Moreover, by Lemmas \ref{lemm:EG} and \ref{lemm:EB} we have $v_{MA}(\mu_q) = v_{MA}(\mu_p)=0$ and $v_{AM}(\mu_q) \leq 0 = v_{AM}(\mu_p)$. So
\begin{align*}
  \underline{v}_\beta(\mu_q) & = \beta v_{MM}(\mu_q) +(1-\beta) v_{AM}(\mu_q) \\
  & \leq \beta v_{MM}(\mu_p) +(1-\beta) v_{AM}(\mu_p) = \underline{v}_\beta(\mu_p)
\end{align*}
while
\begin{align*}
\overline{v}_\beta(\mu_q) & = \beta v_{MA}(\mu_q) +(1-\beta) v_{AA}(\mu_q) \\
& = \beta v_{MA}(\mu_p) +(1-\beta) v_{AA}(\mu_p) = \overline{v}_\beta(\mu_p)
\end{align*}
implying that $(\underline{v}_\beta(\mu_q),\overline{v}_\beta(\mu_q))$ is dominated by $(\underline{v}_\beta(\mu_p),\overline{v}_\beta(\mu_p))$.
\end{proof}

 Thus we can restrict attention to $\{\mu_q : q\in [0,p]\}$. Since for any such $\mu_q$,
 \[\underline{v}_\beta(\mu) = \beta (-q) + (1-\beta)\cdot 0 = -\beta q\]
 while 
 \[\overline{v}_\beta(\mu) = \beta \cdot 0 + (1-\beta)\cdot (q) = (1-\beta) q\] 
we obtain the frontier $\{(-\beta q, (1-\beta)q) :q \in [0,p]\}$ as claimed.

\end{document}